\documentclass[11pt]{article}

\usepackage{amssymb,amsmath,amsfonts,amssymb}
\usepackage{graphics,graphicx,color}
\usepackage{empheq}

\def\@abssec#1{\vspace{.05in}\footnotesize \parindent .2in
{\bf #1. }\ignorespaces}
\graphicspath{{/EPSF/}{Figures/}}
\DeclareGraphicsExtensions{.eps}
\newtheorem{theorem}{Theorem}[section]
\newtheorem{lemma}[theorem]{Lemma}

\newtheorem{remark}[theorem]{Remark}

\def \Rm {\mathbb R}
\def \Nm {\mathbb N}
\def \Cm {\mathbb C}
\def \Zm {\mathbb Z}
\def \Sm {\mathbb S}

\def \Mm {\mathbb M}
\newcommand{\eps}{\varepsilon}

\newcommand{\dsum}{\displaystyle\sum}
\newcommand{\dint}{\displaystyle\int}

\newcommand{\aver}[1]{\langle {#1} \rangle}

\newcommand{\mF}{\mathcal F}
\newcommand{\mH}{\mathcal H}

 \newcommand{\mL}{\mathcal L}
 
 \newcommand{\mS}{\mathcal S}
 
     \newcommand{\rh}{{\rm h}}

\newcommand{\mD}{\mathfrak D}

\newcommand{\fa}{{\mathfrak a}}

\newcommand{\fS}{{\mathfrak S}}
\newcommand{\rS}{{\rm S}}

\newcommand{\cout}[1]{}

\newcommand{\sgn}[1]{\,{\rm sign}(#1)}

\newcommand{\sign}{{\rm sign}}
\newcommand{\ow}{{\rm Op}^w}

\newcommand{\R}{{\rm R}}

\newcommand{\Tr}{{\rm Tr}}
\newcommand{\tdeg}{\widetilde{\deg\, }}
\newcommand{\ind}{{\rm Index}\,}

\newcommand{\horm}{{H\"ormander\ }}
\newcommand{\schr}{{Schr\"odinger\ }}
\newcommand{\tio}{\"o}

\title{Topological charge conservation for continuous insulators}

\author{Guillaume Bal \thanks{Departments of Statistics and Mathematics and CCAM, University of Chicago, Chicago, IL 60637; guillaumebal@uchicago.edu}}

\begin{document}

\maketitle

\begin{abstract}
%\gb{Shorten a bit. As now: (i) classification by domain walls and index of Fredholm operator (ii) Line conductivity and TCC (iii) BEC (iv) Applications\\}

This paper proposes a classification of elliptic (pseudo-)differential Hamiltonians describing topological insulators and superconductors in Euclidean space by means of domain walls. Augmenting a given Hamiltonian by one or several domain walls results in confinement that naturally yields a Fredholm operator, whose index is taken as the topological charge of the system. A Fedosov-\horm formula implementing in Euclidean spaces an Atiyah-Singer index theorem allows for an explicit computation of the index in terms of the symbol of the Fredholm operator. For Hamiltonians admitting an appropriate decomposition in a Clifford algebra, the index is given by the easily computable degree of a naturally associated map.  

A practically important property of topological insulators is the asymmetric transport observed along one-dimensional lines generated by the domain walls. This asymmetry is captured by a line conductivity, a physical observable of the system. We prove that the line conductivity is quantized and given by the index of a second Fredholm operator of Toeplitz type. We also prove a topological charge conservation stating that the two aforementioned indices agree. This result generalizes to higher dimensions and higher-order topological insulators the bulk-edge correspondence of two-dimensional materials. 

We apply this procedure to evaluate the topological charge of several classical examples of (standard and higher-order) topological insulators and superconductors in one, two, and three spatial dimensions.

\end{abstract}


%\tableofcontents

%
%%
\section{Introduction}

This paper considers topological insulators and topological superconductors modeled in the continuum limit by single particle Hamiltonians in the form of (pseudo-)differential operators. See, e.g.,  \cite{bernevig2013topological,liu2010model,RevModPhys.83.1057,sato2017topological,volovik2009universe,witten2016three} for background and details on these materials and topological phases of matter. Topological insulators are characterized by insulating regions separated by interfaces where transport may occur. One of their most important properties practically is that such an interface transport is asymmetric with asymmetry quantized and stable against perturbations. We consider here interfaces modeled by a domain wall, i.e., a scalar function $m(x)$ with a (smooth) interface  (asymptotically a hyperplane) described by the $0-$level set $m^{-1}(0)$ separating two bulk insulating half-spaces where $m>0$ and $m<0$.

A general principle called a bulk-edge correspondence, relates the transport asymmetry along the interface to the bulk properties of the insulator in the regions $\pm m(x)>0$. The bulk-edge correspondence has been derived mathematically in a number of settings for discrete \cite{elbau2002equality,fukui2012bulk,Graf2013,prodan2016bulk,schulz2000simultaneous} as well as continuous Hamiltonians \cite{bourne2017k,bourne2018chern,drouot2021microlocal,
ludewig2020cobordism}; see also \cite{essin2011bulk,volovik2009universe} for a bulk-boundary correspondence and a notion of topological charge conservation similar to the one we will describe in this paper. 
For two-dimensional insulators modeled by differential equations, a bulk-interface correspondence was established in \cite{bal2022topological,QB-NUMTI-2021} by relating a physical observable, an interface conductivity characterizing the interface asymmetry, to the  index of a Fredholm operator and interpreting the integral characterizing the topological index (a Fedosov-\horm formula \eqref{eq:FH} below) as a bulk-difference invariant. These papers consider bounded domain walls in two space dimensions. %The current paper concerns an algebra of operators with unbounded domain walls with an appropriate behavior as the spatial variables $|x|\to\infty$. 

The objective of this paper is to generalize such a correspondence to arbitrary dimensions for a class of elliptic pseudo-differential Hamiltonians with unbounded domain walls. The topology of an insulator is characterized by the index of two different Fredholm operators. The first Fredholm operator is constructed by augmenting the Hamiltonian by confining domain walls. Its index defines the topological charge of the insulator. The second Fredholm operator in Toeplitz form captures the quantized asymmetric transport along {\em one-dimensional} lines. A topological charge conservation result generalizing the bulk-interface correspondence then shows that the two indices agree. The main practical appeal of such a correspondence is that the computation of the first index is typically significantly simpler than that of the second.

\medskip

The construction of the first Fredholm operator is presented in section \ref{sec:local}. Denote by $d$ the spatial dimension and consider operators on the Euclidean space $\Rm^d$. For a given Hamiltonian $H_k$ for $0\leq k<d$, which should be interpreted as confining in the first $k$ spatial  dimensions, we construct an operator $H_{d-1}$ implementing domain walls in the next $d-k-1$ dimensions and finally a Fredholm operator $F=H_{d-1} -i m(x_d)$ where $m(x_d)$ implements a domain wall in the remaining variable $x_d$. The topological charge of $H_k$ is defined as $\ind F={\rm dim\ Ker}\,F-{\rm dim\ Ker}\,F^*$.    The Fredholm operator is {\em local} in the sense that $F=\gamma_1\otimes H_k+\gamma_2$ for $\gamma_1$ a constant matrix and $\gamma_2$ a multiplication operator in the physical variables.  The construction is carried out for Hamiltonians with no symmetry assumption (class A) when $d+k$ is odd and Hamiltonians with a necessary chiral symmetry assumption (class AIII) when $d+k$ is even \cite{bernevig2013topological,prodan2016bulk}. The complex classes A and AIII are the only ones considered in this paper. The operators $H_k$, $H_{d-1}$, and $F$ we consider here are all pseudo-differential operators written in a Weyl quantization. The relevant notation for this paper on pseudo-differential, functional, and semiclassical calculus is recalled in Appendix \ref{sec:notation}. The main guiding choice for the functional setting is that the index of the resulting Fredholm operator $F$ may be computed by means of a Fedosov-\horm formula in Theorem \ref{thm:FH} below. This implementation of an Atiyah-Singer result in Euclidean geometry provides an explicit formula in terms of the symbol of $F$ and hence of that of the original Hamiltonian $H_k$.  The main constraints we impose on the symbol of $F$ are in a nutshell that it be {\em elliptic} and grow to infinity at infinity in phase space in a sufficiently isotropic fashion.

\medskip

While systematic and straightforward, the above classification is unrelated to any physical observable. Such an observable may be assigned to the above intermediate (self-adjoint) operator $H_{d-1}$ as follows. The  operator $H_{d-1}$ confines in $d-1$ directions while allowing transport in the remaining dimension parametrized by $x_d$. The following line conductivity quantifies asymmetric transport in that one dimension.  Let $H:=H_{d-1}$ and $\varphi\in\fS[0,1]$ a smooth non-decreasing switch function and $P=P(x_d) \in \fS[0,1]$ a smooth spatial switch function \footnote{A function $f:\Rm\to\Rm$ is called a switch function $f\in\fS[0,1]$ if $f$ is bounded measurable and there are $x_L$ and $x_R$ in $\Rm$ such that $f(x)=0$ for $x<x_L$ and $f(x)=1$ for $x>x_R$.}. We define $\varphi'(H)$ by functional calculus and then
\begin{equation}\label{eq:sigmaI}
  \sigma_I(H) := \Tr\, i[H,P] \varphi'(H).
\end{equation}
The above conductivity has been used to analyze the bulk-edge correspondence of two dimensional materials in a number of contexts; see, e.g., \cite{bal2022topological,drouot2021microlocal,elbau2002equality,graf2007aspects,Graf2013,prodan2016bulk,QB-NUMTI-2021}. It has the following interpretation. Let $\psi(t)=e^{-i t H} \psi$ be a solution of the \schr equation $i\partial_t\psi(t)=H\psi(t)$ with initial condition $\psi$ and let $P$ be a Heaviside function defined as $P(x_d)=1$ for $x_d>x_0$ while $P(x_d)=0$ for $x_d<x_0$ for some $x_0\in\Rm$. Then $\aver{P}_t=(\psi(t),P\psi(t))_{L^2}$ is interpreted as the amount of signal on the left of $x_0$. Its derivative $\frac{d}{dt} \aver{P}_t = (\psi(t), i[H,P]\psi(t)) = \Tr\, i[H,P]\psi(t)\psi(t)^*$ describes current crossing the interface $x_d=x_0$. Formally replacing the density $\psi(t)\psi(t)^*$ by the spectral density $\varphi'(H)$ heuristically gives the interpretation of $\sigma_I$ as the rate of signal propagating from the left to the right of the hyperplane $x_d=x_0$ per unit time. 

We show in section \ref{sec:sigmaI} that $i[H,P] \varphi'(H)$ is a trace-class operator for the class of operators $H_k$ we consider. Moreover, $\sigma_I$ is related to a Fredholm operator of Toeplitz type as follows. Let $\tilde P\in \fS[0,1]$ be a projector $\tilde P^2=\tilde P$, for instance a Heaviside function. Then $T:=\tilde P U(H)\tilde P_{|{\rm Ran}\tilde P}$ for $U(H)=e^{2\pi i \varphi(H)}$ is a (bounded) Fredholm operator  and $2\pi\sigma_I=\ind T$ showing that $2\pi\sigma_I$ is indeed quantized. 

\medskip

The main result of this paper is the topological charge conservation in Theorem \ref{thm:tcc} of section  \ref{sec:tcc} stating that $2\pi \sigma_I$ is given by the indices of both $F$ and $T$. This topological charge conservation result (see \cite{volovik2009universe}) shows that transport along the line is indeed asymmetric with quantized asymmetry that is stable against any perturbations of the Hamiltonian that preserve the above indices and that this asymmetry may be computed using a (relatively simple) Fedosov-\horm formula. Perturbations may be arbitrary large in a suitable sense and the non-trivial topology may be interpreted as an obstruction to Anderson localization \cite{B-EdgeStates-2018,prodan2016bulk}. 

In two space dimensions, the above topological charge conservation relating the transport asymmetry to the index of $F$ may be interpreted as a bulk-interface correspondence by identifying the explicit integral in \eqref{eq:FH} as a bulk-difference invariant directly related to the bulk properties of $H_{d-1}$ on either side of the interface \cite{bal2022topological,QB-NUMTI-2021}. For three-dimensional materials, the line conductivity or {\em hinge} conductivity \cite{schindler2018higher} provides a simple classification of some higher-order topological insulators.

\medskip

%%%%
While explicit, the computation of the Fedosov-\horm integral in \eqref{eq:FH} below is not always straightforward analytically. However, it significantly simplifies when the symbol $a_k$ of the operator $H_k$ admits a decomposition of the form $ a_k(X) = \rh^k(X) \cdot \Gamma_k$ where  $\rh^k(X)$ is a vector field on phase space $X\in\Rm^{2d}$ of dimension $d+k$ and $\Gamma_k$ is a vector of Clifford matrices satisfying appropriate anti-commutation properties; see \eqref{eq:clifcomm} in section \ref{sec:deg}.  We show in that section that we may define an appropriate degree of the map $\rh^k(X)$. Theorem \ref{thm:tccp} then states that the topological charge of $H_k$ is, up to a sign depending on choices of orientation for the Clifford matrices and the phase space variables $X$, equal to the degree of $\rh^k$ defined for any regular value $y_0$ of $\rh^k$ in the vicinity of $0$ (see section \ref{sec:deg}) by
\begin{equation}\label{eq:degexplicit}
   \deg \rh^k  = \dsum_{\zeta \in (\rh^k)^{-1}(y_0)} \sign \det J(\zeta)
\end{equation}
with $J$ the non-degenerate Jacobian matrix of the map $\rh^k$.
% the decomposition \gb{Decomposition below may be summarized here and presented in more detail in section 5.}
%\begin{equation}\label{eq:clifford} 
%  a_k(X) = \rh^k(X) \cdot \Gamma_k.
%\end{equation}
%Here, $\rh^k(X)$ is a vector field on phase space $X\in\Rm^{2d}$ of dimension $d+k$. 
% Let $\kappa:=\kappa_k=\lfloor \frac{d+k}2 \rfloor$ and $n_k=2^{\kappa_k}$ for $0\leq k\leq d$.
%The matrices $\Gamma_k=(\gamma_\kappa^j)_j$ for $1\leq j\leq d+k$ are constructed so that they satisfy the defining commutation relations of  the Clifford algebra  ${\rm Cl}_{n_k}(\Cm)$:
%\begin{equation}\label{eq:cliffordcomm}
%    \gamma_\kappa^i\gamma_\kappa^j+\gamma_\kappa^j\gamma_\kappa^i=2\delta_{ij} I_{n_k}, \quad 1\leq i,j\leq d+k.
%\end{equation}
%We show in section \ref{sec:deg}  that for $y_0$ (typically $y_0=0$) a regular value of $\rh^k$ in the vicinity of $0$, we may define the degree of $\rh^k$ as 
%\begin{equation}\label{eq:degexplicit}
%   \deg \rh^k  = \dsum_{\zeta \in (\rh^k)^{-1}(y_0)} \sign \det J(\zeta)
%\end{equation}
%with $J$ the non-degenerate Jacobian matrix of the map $\rh^k$.  The main result of section \ref{sec:deg} is Theorem \ref{thm:tccp} stating that the topological charge of $H_k$ is, up to a sign depending on choices of orientation for the Clifford matrices and the phase space variables $X$, equal to the degree of $\rh^k$. 
%
This provides a particularly simple means to compute the topological charge of several of the operators that appear in the analysis of topological insulators and topological superconductors.

\medskip

In section\ref{sec:appli} , we apply the above results to the computation of the topological charge of several operators that appear in the modeling of topological insulators and superconductors, higher-order topological insulators, and topological models of fluid waves. 

\medskip

The above classification based on domain walls shares similarities with other classification mechanisms. Topological phases of matter are characterized by non-trivial topologies of Hamiltonians in dual, Fourier, variables \cite{bernevig2013topological,prodan2016bulk,volovik2009universe}. This non-trivial topology may be tested in several ways. Linear response theory in two dimensions tests a domain by applying a linear electric field in one direction and assessing the resulting transport in the transverse direction leading to the notion of Hall conductivity. While physically different, adding a domain wall to $H_0$ in two dimensions and testing asymmetric transport by $\sigma_I(H_1)$ in the transverse direction plays a similar classifying role. Toeplitz operators of the form $T=PUP$ with $P$ a projector and $U$ unitary also appear naturally in the classification of topological insulators by means of Fredholm modules \cite{prodan2016bulk}; see also \cite{B-BulkInterface-2018} for an application to (regularized) Dirac operators in Euclidean space.  The main advantage in the classification based on $\ind F$ is the explicit Fedosov-\horm formula it satisfies. The latter is also used in a different context of operators in Euclidean space by \cite{bott1978some,callias1978axial} to test the topology of a physical potential with appropriate behavior at infinity in the physical variables using a Dirac operator. 

As in \cite{bal2022topological,QB-NUMTI-2021}, we do not classify bulk phases but rather transitions from one phase to the other. Bulk phases are not defined for many unperturbed operators \cite{B-BulkInterface-2018,bal2022topological} showing that it is easier to describe phase transitions rather than absolute phases.

\section{Fredholm operator and topological charge}
\label{sec:local}
%%
%

%\gb{Could start with Dirac operator in several dimensions: to remove from introduction section.}

%%%%%
\paragraph{\bf Classification of Dirac operators.}
To illustrate how the topology of Hamiltonians is tested by domain walls, we present the constructions for Dirac operators, which are generic models for band crossings \cite{bernevig2013topological,drouot2021microlocal,FLW-ES-2015,sato2017topological} and arguably the simplest models for topological phases of matter.

Consider first a one-dimensional material and asymmetric transport modeled by the Hamiltonian $H_0=D_x$ with $D_x=-i\partial_x$, which may be seen as a self-adjoint operator on $L^2(\Rm)$ with domain $\mD(H_0)=H^1(\Rm)$.  We introduce the operator $F=D_x-ix=-i\fa$ with $\fa=\partial_x+x$ an annihilation operator. The term $-ix$ is interpreted as a domain wall confining the particle to the vicinity of $x=0$. The operator $F$ is now a Fredholm operator from its domain of definition $\mD(F)=\{f\in L^2(\Rm); f' \in L^2(\Rm) \mbox{ and } xf \in L^2(\Rm)\}$ to $L^2(\Rm)$. Moreover, we verify that $\ind F=1$ with kernel associated to the eigenfunction $e^{-\frac12 x^2}$. 
The line conductivity $\sigma_I(H_0)$ in \eqref{eq:sigmaI} describes the asymmetric transport associated to $H_0$. 
%We define the topological charge of the Hamiltonian $H_0$  to be the index of $F$.

Consider next the two-dimensional version of the above example, where $H_0=D_1\sigma_1+D_2\sigma_2$ with $D_j=-i\partial_j$ for $j=1,2$ and $\sigma_{1,2,3}$ are the standard Pauli matrices. The operator $H_0$ appears as a generic low-energy description of energy band crossings and is ubiquitous in works on topological insulators. We classify $H_0$ by augmenting it with a domain wall along one direction and assessing the resulting asymmetric transport in the transverse direction. We implement a domain wall along the first variable by introducing $H_1=H_0 + x_1 \sigma_3$. This models insulating regions when $\pm x_1>0$. The line conductivity $\sigma_I(H_1)$ in \eqref{eq:sigmaI} describes the asymmetric transport of $H_1$. Associated to $H_0$ and $H_1$ is the operator $F=H_1-ix_2$. This is again a Fredholm operator from its domain of definition to $L^2(\Rm^2)\otimes \Cm^2$ and we verify that $\ind F=1$, which defines the topological charge associated to $H_0$ (and $H_1$). The kernel of $F$ has for eigenfunction the spinor $e^{-\frac12 |x|^2} (1,i)^t$. 

This construction generalizes to higher dimensions in a straightforward way except for the fact that the construction of the domain walls requires additional degrees of freedom as dimension increases. Consider in $\Rm^3$ the Weyl Hamiltonian $H_0=D_1\sigma_1+D_2\sigma_2+D_3\sigma_3$. As an operator acting on spinors in $\Cm^2$, the latter operator is stable against gap opening by domain walls \cite{bernevig2013topological}. We therefore introduce the operator $H_1=\sigma_1\otimes H_0+ \sigma_2\otimes I_2 x_1$ with a domain wall in the first direction but now acting on spinors in $\Cm^4$. The operator $H_1$ thus admits surface states concentrated in the vicinity of $x_1=0$, as does the operator $H_0$ in the two-dimensional setting. Its topology is then characterized by asymmetric transport in the third dimension after a second domain wall in the $x_2$ direction is introduced: $H_2=H_1+\sigma_3\otimes I_2 x_2$. Asymmetric transport along the line $x_1=x_2=0$ is again described by a line conductivity $\sigma_I(H_2)$. The topological charge of $H_0$ (and that of $H_k$ for $k=1,2$) is then defined as the index of the Fredholm operator (from its domain of definition to $L^2(\Rm^3)\otimes \Cm^4$) $F=H_2-ix_3$. We verify (and will show in greater generality) that the topological charge of $H_0$ is $\ind F=-1$, with a sign change here reflecting the fact that indices depend on the orientation of the Clifford matrices used to construct the operators $H_j$ as well as the orientation of the domain walls. The kernel of $F^*$ has for eigenfunction the spinor $e^{-\frac12 |x|^2}(1,-1,-1,-1)^t$.

%%%%%
\paragraph{Pseudo-differential elliptic operators.} %%% \gb{This part, which was redundant, should mostly go. Useful notation from appendix may be recalled here. In particular the type of operators considered here and some examples. Bizarre place here: operators and examples should go before Fredholm construction or possibly in the application section.}

Consider a spatial dimension $d\geq1$ and operators defined on functions of the Euclidean space $\Rm^d$.  We denote by $\xi\in\Rm^d$ the dual (Fourier) variable and $X=(x,\xi)\in\Rm^{2d}$ the phase space variable.  We generalize the above construction to an appropriate algebras of pseudo-differential  operators (PDO) with Weyl quantization defined as
\begin{equation}\label{eq:weyl}
  H \psi (x) = \ow a\  \psi (x) := \dfrac{1}{(2\pi)^d} \dint_{\Rm^{2d}} e^{i\xi\cdot (x-y)} a(\frac{x+y}2,\xi)  \psi(y)dy d\xi,
\end{equation}
for $a(X)$ a matrix-valued symbol in $\Mm(n)$ and $\psi(x)$ a spinor with values in $\Cm^n$. 

We will start with an operator $H_k=\ow a_k$ with $a_k$ a matrix-valued symbol in $\Mm(n_k)$ interpreted as confining in the first $k$ variables and aim to construct the operators  $H:=H_{d-1}=\ow a_{d-1}$ and $F=\ow a$. To do so, we introduce the following notation and functional setting.
 
We decompose the spatial variables $x=(x'_k,x''_k)$ with $x'_k\in\Rm^k$ and $x''_k\in\Rm^{d-k}$. We use the notation $\aver{y}=\sqrt{1+|y|^2}$ and $\aver{y_1,y_2}=\sqrt{1+|y_1|^2+|y_2|^2}$ and define the weights
\begin{equation}\label{eq:weights1}
w_k(X)=\aver{x_k',\xi} . % \quad \mbox{ for } \quad m\in \Nm_+.
\end{equation}
For a given spinor dimension $n=n_k$ with $0\leq k\leq d$, and an order $m>0$, we denote by $S^m_k=S^m_k[n_k]$ the class of symbols $a_k$ such that for each $d-$dimensional multi-indices $\alpha$ and $\beta$, there is a constant $C_{\alpha,\beta}$ such that for each component $b$ of $a_k\in \Mm(n_k)$, we have
\begin{equation}\label{eq:Sk}
   \aver{x}^{|\alpha|} \aver{\xi}^{|\beta|} |\partial^\alpha_x\partial^\beta_\xi b(X)| \leq C_{\alpha,\beta} w_k^m(X),\qquad \forall X\in\Rm^{2d}.
\end{equation}
We also define the space of symbols $\tilde S^m$ as $S^m_d$ above but acting on spinors of (lower) dimension $n_{d-1}$ instead of $n_d$.  Here $m$ is the order of the operator. 

For the two-dimensional  Dirac operator, we find $m=1$, $H_0= \ow a_0$ for $a_0=\xi_1\sigma_1+\xi_2\sigma_2$ while $H_1=\ow a_1$ for $a_1=a_0+x_1\sigma_3$ and $F=\ow a$ for $a=a_1-ix_2$. For $n_0=n_1=2$, we observe that $a_j\in S^1_j$ for $j=0,1$, while $a\in \tilde S^1$.

Associated to the spaces of symbols $S^m_k$ and $\tilde S^m$ are Hilbert spaces $\mH^m_k$ and $\tilde \mH^m$;  see \eqref{eq:Hspaces} in Appendix \ref{sec:notation}. These spaces are constructed so that for $a_k\in S^m_k$, we have that $\ow a_k$ defined in \eqref{eq:weyl} maps $\mH^m_k$ to $\mH^0_k=L^2(\Rm^d)\otimes \Mm(n_k)$ continuously.

The main assumption we impose on $H_k$, beyond the Hermitian symmetry $a_k=a_k^*$ and a chiral symmetry (see \eqref{eq:chsym} below)  when $d+k$ is even, is that it be {\em elliptic} as an operator from  $\mH^m_k$ to $\mH^0_k$. This is the case when there are constants $C_{1,2}>0$ such that 
\begin{equation}\label{eq:ellip}
  | \det a_k(X) | ^{\frac 1{n_k}} \geq C_1w_k^m(X) - C_2,\quad \forall X\in\Rm^{2d}.
\end{equation}
In other words,  all eigenvalues of the Hermitian matrix $a_k(X)$ are bounded away from $0$ by at least $Cw_k^m(X)$ for $X$ outside of a compact set in $\Rm^{2d}$. Since $a_k\in S^m_k$, all (positive and negative) eigenvalues of $a_k$ are of order $w_k^m(X)$ away from a compact set.  

We denote by $ES^m_k$ the {\em elliptic symbols} in $S^m_k$ and $E\tilde S^m$ the elliptic symbols in $\tilde S^m$. For $a_k\in ES^m_k$, we obtain that $H_k$ is an unbounded self-adjoint operator with domain of definition $\mH^m_k$ while for $a\in E\tilde S^m$, we obtain that $F=\ow a$ is an unbounded operator with domain of definition $\tilde \mH^m$; see Appendix \ref{sec:notation}.

\medskip

\paragraph{Classification by domain walls.}
We start from an elliptic (self-adjoint) operator $H_k=\ow a_k$ for $a_k\in ES^m_k$ for $m>0$. By ellipticity assumption, $H_k$ is an unbounded self-adjoint operator with domain $\mD(H_k)=\mH_k^m$ and acts of spinors in $\Cm^{n_k}$. The ellipticity of $H_k$ and the construction of the weight $w_k(X)=\aver{x_k',\xi}$ imply that the first $k$ variables parametrized by $x_k'$ are confined in the sense that $|a_k|$ is large for $|x_k'|$ large.

To be non-trivial topologically, the operator $H_k$ needs to satisfy a chiral symmetry when $d+k$ is even (complex class AIII \cite{prodan2016bulk}). When $d+k$ is odd, then $H_k$ is in the complex class A with no symmetry imposed beyond the Hermitian structure.

Assume first that $d+k$ is even with $k\leq d-2$. Recall that $\sigma_{1,2,3}$ are the Pauli matrices, or more generally any set of Hermitian $2\times2$ matrices such that $\sigma_i\sigma_j+\sigma_j\sigma_i=2\delta_{ij}$ and $\sigma_1\sigma_2=i\sigma_3$. Using the notation $\sigma_\pm=\frac12(\sigma_1\pm i\sigma_2)$, the chiral symmetry takes in a suitable basis the following form:
\begin{equation}\label{eq:chsym}
  H_k =  \begin{pmatrix} 0 & F_k \\ F_k^* & 0 \end{pmatrix} = \sigma_- \otimes F_k^* + \sigma_+\otimes F_k.
\end{equation}
%.  \tb{Do we fix $\sigma_i$ and rotate $H_k$ or fix $H_k$ and adapt $\sigma_j$? Leaving $\sigma_j$ arbitrary might be a good choice. We could then use $\tau_{1,2,3}$ for the standard Pauli matrices.}
 We next introduce the domain wall 
\begin{equation}\label{eq:dwall}
m_{k}(X):=\aver{x_{k}}^{m-1} x_{k}.
\end{equation}
They are constructed to have the same asymptotic homogeneity of order $m$ as the Hamiltonian $H_k$. 
We then define the new spinor dimension $n_{k+1}=n_k$  and the augmented Hamiltonian
\begin{equation}\label{eq:augeven}
  H_{k+1} := H_k + m_{k+1} \sigma_3 \otimes I.
\end{equation}
This implements a domain wall in the variable $x_{k+1}$.

Assume now $d+k$ odd.  Then we define the new spinor dimension $n_{k+1}=2n_k$ and the augmented Hamiltonian
\begin{equation}\label{eq:augodd}
 H_{k+1} := \sigma_1\otimes H_k+m_{k+1} \sigma_2 \otimes I  = \sigma_-\otimes F^*_{k+1} + \sigma_+\otimes F_{k+1},\quad
  F_{k+1} = H_{k} -i m_{k+1}.
\end{equation}
The operator $H_{k+1}$ satisfies a chiral symmetry of the form \eqref{eq:chsym}, as requested since $d+k+1$ is now even.

We denote by $a_{k+1}$ the symbol of $H_{k+1}=\ow a_{k+1}$ and observe that $a_{k+1}=a_k+m_{k+1}\sigma_3\otimes I$ when $d+k$ is even and $a_{k+1}= \sigma_1\otimes a_k+m_{k+1}\sigma_2\otimes I$ when $d+k$ is odd. 

The procedure is iterated until $H_d$ has been constructed. Note that $a_{k+2}(X)\in \Mm(2n_k)$ with dimension of the spinor space on which the matrices act that doubles every time $k$ is raised to $k+2$. Since $2d$ is even, $H_{d}=\sigma_-\otimes F^* + \sigma_+ \otimes F$ for an operator $F=F_{d}=H_{d-1}-im_d =: \ow a$, or equivalently $a=a_{d-1}-im_d$.% Since $F$ is a central operator of the framework, we denote by $a$ its symbol. 

For $0 < l \leq d-k$, the intermediate Hamiltonians all have the form
\[
  H_{k+l} =  \gamma_0 \otimes H_k + \mu\cdot\gamma\otimes I_{n_k}
\]
where for some integer $p=p(l,k)$ and for some matrices $\gamma_j$ such that $\{\gamma_i,\gamma_j\}:=\gamma_i\gamma_j+\gamma_j\gamma_i=0$ for all $i\not=j$ in $\{0,\ldots,l\}$, we have
\[
 \mu=(m_{k+1},\ldots, m_{k+l}), \quad \gamma_0=\sigma_1^{\otimes p} ,\quad \gamma=(\gamma_{1},\ldots,\gamma_{l} ).
\]

We now show that all operators $H_l$ are elliptic and that $H_d$ and $F$ are Fredholm operators.
\begin{lemma}\label{lem:ellj}
Let $a_k\in ES_k^m$% for $g\in \{g^{i}, g^s\}$ 
and $H_k=\ow a_k$ satisfying the chiral symmetry \eqref{eq:chsym} when $d+k$ is even. Then $a_j\in ES_j^m$ for all $k\leq j\leq d$ and $a\in E\tilde S^m$.
\end{lemma}
\begin{proof}
  Let $a_{k+l}$ be the symbol of $H_{k+l}=\ow a_{k+l}$ for $0<l\leq d-k$. By construction and commutativity $\{\gamma_i,\gamma_j\}=0$ as recalled above, we obtain that $a_{k+l}^2=I\otimes a_k^2 + \sum_{j=1}^l m_{k+l}^2 \otimes I$ with $I$ identity matrices with appropriate dimensions. This shows that $a_{k+l}$ satisfies the ellipticity condition \eqref{eq:ellip} for the weight $w_{k+l}(X)$. The decay properties for derivatives of $a_{k+l}$ in \eqref{eq:Sk} with $k$ replaced by $k+l$ follow from the corresponding properties for $a_{k}$. That $a\in E\tilde S^m$ comes from the corresponding result for $a_d$ and the construction of $F$.
 %TBD; should be relatively straightforward.%\tb{Above metrics have to change as written in a remark earlier.}
\end{proof}

Let $\Lambda= \sqrt{-\Delta +|x|^2+1} = \sqrt{\ow \aver{x,\xi}^2}$ be an elliptic self-adjoint operator, which by construction, maps $\mH(w_{d-1}^m)$ to $\mH(1)=L^2(\Rm^d)$ \cite{nicola2011global}.
\begin{lemma}\label{lem:fredholm}
%Let $g\in \{g^i,g^s\}$.
The above operators $H_d$ and $F$ are Fredholm operators from $\mH_d^m$ to $\mH_d^0$ and  $\tilde\mH^m$ to $\tilde\mH^0$, respectively. Equivalently $\Lambda^{-m}H_d$ and $\Lambda^{-m}F$ are Fredholm operators on $\mH_d^0$ and $\tilde\mH^0$, respectively.
\end{lemma}
\begin{proof}
This is  \cite[Theorem 18.6.6]{H-III-SP-94}; see also \cite{nicola2011global} since the Planck function $h^s(X)$ tends to $0$ as $|X|\to\infty$.
%This is essentially in [NR] as well as in  Theorem 18.6.6 in H3 and the fact that $h(x,\xi)$ goes to $0$ as $(x,\xi)\to\infty$; this is applied to the residuals $R_j$ whose symbols belong to $S(h,g^s)$. This also holds for the smaller class  with $g^{i}$ obviously. 
\end{proof}

The construction of $H_{l}$ implements $l-k$ domain walls to test the topology of the operator $H_k$.  When $l=d$, the operator $H_d$ has $d$ confined variables and is as we saw a Fredholm operator, i.e., an operator that is invertible modulo compact operators \cite[Chapter 19]{H-III-SP-94}. The operator $H_d$ is self-adjoint and so its index vanishes. However, it satisfies the chiral symmetry \eqref{eq:chsym} and the corresponding operator $F=H_{d-1}-im_d$ is also Fredholm. Its index may not vanish and provides a definition of the topological charge of $H_k$. 

The intermediate operator $H_{d-1}$ is physically relevant with $d-1$ confined spatial variables (close to $x'_{d-1}=0$) and transport allowed along the direction $x''_{d-1}=x_d$. As we show in section \ref{sec:sigmaI}, this transport is asymmetric and quantized by the topological charge of $H_k$. 

\paragraph{Topological Charge and Integral Formulation.}
We next apply \cite[Theorem 19.3.1']{H-III-SP-94} (see also \cite{fedosov1970direct}) to the operator $\Lambda^{-m}F$ to obtain that the index of $F$, which equals that of $\Lambda^{-m}F$ since $\Lambda^{-m}$ has trivial index, is given by the following Fedosov-\horm formula
\begin{theorem}\label{thm:FH}
 For the above operator $F=\ow a$, we have
\begin{equation}\label{eq:FH}
   {\rm Index}\ F =  - \dfrac{(d-1)!}{(2\pi i)^d (2d-1)!}\dint_{\Sm_R^{2d-1}} {\rm tr} (a^{-1} da)^{2d-1}.
\end{equation}
Here, $R$ is a sufficiently large constant so that $a$ is invertible outside of the ball of radius $R$ and the orientation of $\Rm^{2d}$ and that induced on $\Sm_R^{2d-1}$ is chosen so that $d\xi_1 \wedge dx_1 \wedge \ldots \wedge d\xi_d \wedge dx_d >0$. 
\end{theorem}
Note that  \cite[Theorem 19.3.1]{H-III-SP-94} applies to the smaller class of symbols $S(M,g^i)$ and the above theorem comes from the approximation of symbols in $S(M,g^s)$ by symbols in $S(M,g^i)$ as described in  \cite[Lemma 19.3.3]{H-III-SP-94}. We will use a similar approximation in Lemma \ref{lem:contiso} to prove the topological charge conservation in Theorem \ref{thm:tcc}.

Note also that \cite[Theorem 19.3.1']{H-III-SP-94} applies to $\Lambda^{-m}F$. However, the index is independent of $t\in[0,1]$ for $\Lambda^{-tm}F$ and the corresponding symbols $a_t$ are uniformly invertible for $|X|\geq R$. The formula \eqref{eq:FH} is then seen as the degree of the map $a_t$ from the sphere $\Sm_R^{2d-1}$ to $GL(n_{d-1};\Cm)$, which as such is a continuous integer and hence independent of $t\in [0,1]$; see, e.g., \cite{bott1978some}. This proves \eqref{eq:FH} for $F$ as an operator from $\tilde \mH^m$ to $\tilde \mH^0$.

Following \cite{volovik2009universe}, we call this index the topological charge of  $H_k$ and $F$.

%This is Hormander's theorem including the extension to $g^s$.  Note that $a$ may be written in terms of the symbol $a_k$ of $H_k$. We will get a simpler expression for the above index for operators expressible in a Clifford algebra.

%
%
\section{Physical observable and Toeplitz operator}
\label{sec:sigmaI}
We now present a second topological classification based on the physical observable given by the line conductivity $\sigma_I(H_{d-1})$ in \eqref{eq:sigmaI}. The line conductivity enjoys stability properties that can often be established directly from its definition as a trace \cite{drouot2021microlocal,elbau2002equality,Graf2013,prodan2016bulk,QB-NUMTI-2021}.  Here, we follow \cite{bal2022topological} and relate the conductivity to the index of the Toeplitz operator $T:=\tilde PU(H) \tilde P _{{\rm Ran} \tilde P}$ for $\tilde P$ an orthogonal projector in $\fS[0,1]$. We prove below that $T$ is a Fredholm operator from ${\rm Ran} \tilde P \subset \tilde \mH^0$ to itself, or equivalently that $\tilde PU(H) \tilde P + (I-\tilde P)$ is a Fredhom operator on $\tilde \mH^0$.

%We wish to implement the notion of (quantum) Hall transport. Let $H_k$ be given and elliptic in $k$ spatial variables (heuristically meaning that the last $k$ variables are confined in the vicinity of the origin as described in the preceding section). Let us choose a frame for the remaining $d-k$ variables. We then construct $d-k-1$ domain walls as in the preceding section. This gives $H_{d-1}$. The Hall effect then consists in estimating the asymmetric transport in the remaining direction parametrized by $x_d$. We assume here that all domain walls are Cartesian, or in other words that there is an appropriate global change of variables on $\Rm^d$ after which the Hamiltonian is elliptic in the last $k$ variables. 
%
%We then associate the physical Hall transport observable
%\[
%  \sigma_I = \Tr\ i[H,P]\varphi'(H).
%\]
%To simplify notation, we define $H:=H_{d-1}$. $P=P(x_d)$ is a smooth switch function in $\fS[0,1]$ and $\varphi$ is also a smooth switch function in $\fS[0,1]$. We then interpret $\varphi'(H)$ as a density of states. Finally, $i[H,P]$ corresponds to the rate of change of the observable $P$, which quantifies how much signal moves from the left to the right. We will obtain that $2\pi\sigma_I$ is quantized and an integer.

Let $a_k\in ES_k^m$ so that by Lemma \ref{lem:ellj}, $a_{d-1}\in ES_{d-1}^m$ while $a=a_{d-1}-ix_d \in E\tilde S^m$. We denote by $H=H_{d-1}=\ow a_{d-1}$. % For the rest of the section, $g\in \{g^s,g^i\}$. 
Let $U(H)=e^{i2\pi \varphi(H)}$ with $\varphi\in C^\infty\fS[0,1]$ (the set of $C^\infty$ switch functions) while $W(H)=U(H)-I$.
\begin{lemma}\label{lem:indexT}
  Let $a_k\in ES_k^m$ and $P\in C^\infty\fS[0,1]$. Then $[P,W(H)]$ and $[P,H]\varphi'(H)$ are trace-class operators with symbols in $S_{d}^{-\infty}$. When $\tilde P\in \fS[0,1]$ is an orthogonal projector, then  $T:=\tilde PU(H) \tilde P _{{\rm Ran} \tilde P}$ is a Fredholm operator on ${\rm Ran} \tilde P \subset\tilde \mH^0$ with index given by ${\rm Tr}[U(H),\tilde P]U^*(H) = {\rm Tr}[U(H), P]U^*(H)$.  All above operator traces may be computed by integrating the Schwartz kernel of the operator along the diagonal.
\end{lemma} 
\begin{proof}
The proof are similar to those of \cite[Lemmas 4.1\&4.2 \& Proposition 4.3]{bal2022topological}. Since the setting and notation slightly differ, we provide a reasonably detailed derivation.

We recall the definition of the symbol class $S^0(M)$ in \eqref{eq:Sjh} with $j=0$ for $M$ an order function (with $h-$independent symbols thus corresponding to a choice of Euclidean metric $g_X=dx^2+d\xi^2$). By composition calculus \cite[Chapter 7]{dimassi1999spectral}, for any $A=\ow a$ with $a\in S^0(M)$, then the decomposition $[A,P]=(1-\chi(x_d))A\chi(x_d)-\chi(x_d)A(1-\chi(x_d))$ for $\chi(x_d)$ a smooth function equal $1$ for $x>1$ and $0$ for $x<-1$ shows that $[A,P]$ has symbol in $S^0(M\aver{x_d}^{-\infty})$ (i.e., in $S^0(M\aver{x_d}^{-N})$ for each $N\in\Nm$). By assumption on $a_k$ and using the functional calculus result in Lemma \ref{lem:fccalc}, we obtain for $\phi\in C^\infty_c(\Rm)$ that $\phi(H)\in S^0(\aver{x'_{d-1},\xi}^{-\infty})$, which is larger that $S_{d-1}^{-\infty}$. To simplify notation, we use the same notation for $S^0(M)$ and $S^0(M)\otimes \Mm(n)$ for any $n$. Therefore, from the above, $[\phi(H),P]$ and $[H,P]\phi(H)$ as well as $H^p[H^q,P]\phi(H)$ for $p,q\in\Nm$ all have symbols in $S^0(\aver{X}^{-\infty})$. We use this with $\phi=\varphi'$ and $\phi=W$. With additional effort, we verify that all symbols are in $S_d^{-\infty}$, although this is not necessary for the rest of the proof and so we leave the details to the reader.

It is then clear \cite[Theorems 9.3\&9.4]{dimassi1999spectral} that $[P,W(H)]$ and $[P,H]\varphi'(H)$ are trace-class operators with traces given as the integral of their Schwartz kernel along the diagonal, or equivalently as the phase-space integral of their symbol. Applying the latter directly yields that $\Tr\,[\phi(H),P]=0$ for instance.

\medskip

Let now $\tilde P\in \fS[0,1]$ a switch function that is not necessarily smooth (and so that for instance $\tilde P^2=\tilde P$). Up to possible rescaling of the variable $x_d$, we assume that $\tilde P(x_d)=0$ for $x_d<-1$ and $\tilde P(x_d)=1$ for $x_d>1$. The above composition calculus no longer applies and a more detailed decomposition of $[\phi(H),\tilde P]$ is necessary. We follow \cite[Lemma 4.2]{bal2022topological}. Let $\tilde w(x,\tilde x)$ be the Schwartz kernel of $\phi(H)$ and introduce $w$ such that $w(\frac12(x+x'),x-x',y,y')=\tilde w(y,x;y',x')$. Since $\phi(H)\in S^0(\aver{x'_{d-1},\xi}^{-\infty})$, we obtain that $w$ is smooth in all variables and rapidly decaying in the last three variables. 

Let $h_m$ be an orthonormal basis of $L^2(\Rm)$, for instance the Hermite functions. We decompose $w(x_1,x_2,y,y') =\sum_{m,n} w_{mn}(x_1,x_2) h_m(y) h_n(y')$ and $\phi(H)=\sum_{m,n} \phi_{mn}$. By assumption $\sum_{m,n} m^\alpha n^\beta |w_{mn}|$ is bounded uniformly in $(x_1,x_2)$.  It is thus sufficient to show that $[\phi,\tilde P]:=[\phi_{mn},\tilde P]$  (dropping the indices $(m,n)$, also in $w:=w_{mn}$) is trace-class with trace-norms summable in $(m,n)$ since $w_{mn}$ decays more rapidly than $\aver{m}^{-2}\aver{n}^{-2}$, say.

Now $w(\frac12(x+x'),x-x')$ is smooth in both (one-dimensional) variables and rapidly decaying in $x-x'$.  Let $k(x,x')$ be the Schwartz kernel of the corresponding $[\phi,\tilde P]$. We introduce a smooth partition of unity $1=\sum_{j\in\Zm}\chi_j(x)$ such that $\chi_0(x)$ has support in $(-2-\eta,2+\eta)$ and equals $1$ on $(-2+\eta,2-\eta)$ for $0<\eta<\frac12$; $\chi_j(x)$ has support in  $(1+j-\eta,2+j+\eta)$ and equals $1$ on $(1+j+\eta,2+j-\eta)$ while $\chi_{-j}(x)=\chi_j(-x)$ for $j\geq1$. The Schwartz kernel of $[\phi,\tilde P]$ is therefore  given by
\[
  k(x,x') = \dsum_{i,j\in\Zm} k_{ij}(x,x') , \quad 
   k_{ij}(x,x') = (\tilde P(x')-\tilde P(x))  w(\frac{x+x'}2,x-x') \chi_i(x) \chi_j(x').
\]
We now show that each $k_{ij}(x,x')$ is the kernel of a trace-class operator with trace-class norm summable in $(i,j)$. 

By assumption on $\tilde P$, we observe that $k_{ij}(x,x')=0$ when $0<i,j$ and $i,j<0$. By symmetry, it remains to consider the cases $\{i=0$ and $j\geq0\}$ as well as $\{i\leq-1$ and $j\geq1\}$. 
Assume first $i=0$ and $0\leq j\leq 5$. These contributions are of the form
\[
   (\tilde P(x')-\tilde P(x)) w(\frac{x+x'}2,x-x')  \phi_1(x)\phi_2(x')
\]
with $\phi_p$ compactly supported for $p=1,2$. Following \cite[Section 9]{dimassi1999spectral}, this is decomposed as
\[
   \dint_{\Rm^2} \hat w(\xi,\zeta) e^{i\frac{x+x'}2\zeta} e^{i(x-x')\xi} d\xi d\zeta  (\tilde P(x')-\tilde P(x)) \phi_1(x)\phi_2(x').
\]
This may be seen as a superposition in $(\xi,\zeta)$ of rank-one operators with traces uniformly bounded in $(\xi,\zeta)$ since $\phi_j$ and $\tilde P$ are bounded. By regularity assumptions on $w$, then $\hat w(\xi,\zeta)\in L^1(\Rm^2)\otimes \Mm(n)$. For instance by decomposing $\hat w$ in a basis of Hermite functions, the above trace is well approximated by that of finite rank operators so that all traces are computed as integrals of Schwartz kernels along the diagonal $x=x'$. 

Let us now consider the cases $i=0$ and $j\geq6$ or $i\leq-1$ and $j\geq1$. We observe for $x\not=x'$ that $e^{i(x-x')\xi}=\frac{1}{[i(x-x')]^N}\partial^N_\xi e^{i(x-x')\xi}$. We then find after integration by parts that 
\[
  k_{ij}(x,x') = i^3 \dint_{\Rm^2} \partial^3_\xi \hat w(\zeta,\xi) e^{i\frac{x+x'}2\zeta} e^{i(x-x')\xi} d\xi d\zeta  (\tilde P(x')-\tilde P(x)) \dfrac{\chi_i(x)\chi_j(x')}{(x-x')^3}.
\]
By assumption on $w$, we have that $\partial^3_\xi \hat w(\zeta,\xi)$ is also integrable.

In all cases, we observe that on the support of $\chi_i(x)\chi_j(x')$, we have $3<x'-x<5$. Therefore $x-x'=q+\tilde x-\tilde x'$ for $q\geq3$ an integer and $|\tilde x-\tilde x'|\leq2$. Thus,
\[
  \dfrac{1}{(x-x')^3} =\dfrac{1}{(q+ \frac{\tilde x-\tilde x'}q)^3} = \frac {1}{q^3}\dsum_{m_1,m_2,m_3\geq0}  \Big( \frac{\tilde x'-\tilde x}q\Big)^{m_1+m_2+m_3}.
\]
This absolutely convergent sum is well-approximated by finite sums, which all give rise to finite-rank operators (degenerate Schwartz kernels). This shows, using the regularity of $w$, that $k_{ij}$ is the kernel of an operator given as a limit in trace-class norm of trace-class operators with a trace norm bounded by $C(|i|+j)^{-3}$ when $i<0<j$ and bounded by $j^{-3}$ when $i=0$ and $j\geq6$. Since the latter bound is summable in $(i,j)$, this shows that $[\phi(H),\tilde P]$ is trace-class as limit of finite rank trace-class operators. Since the traces of the latter are given by the integral of their kernels along the diagonal, this property also holds for $[\phi(H),\tilde P]$. Since the kernel of that operator vanishes along the diagonal, we find that ${\rm Tr}[\phi(H),\tilde P]=0$.

\medskip

Let now $B$ be a bounded operator on $L^2(\Rm^{2d})\otimes \Mm(n)$. We just showed that $[U(H),\tilde P]=[W(H),\tilde P]$ could be decomposed as a trace-class limit of sums of rank-one operators. Let $K_{12}$ be such an operator. Then $K_{12}B$ is also a rank-one operator with trace norm increased at most by $\|B\|$. This shows that $[U(H),\tilde P]B$ is also trace-class. 

Let us finally consider the operators $(P-\tilde P)\phi(H)$ and $\phi(H)(P-\tilde P)$ for $P$ a smooth switch function. Since $\tilde P-P$ now has compact support, we can modify the above partition of unity with $\chi_0=1$ on the support of $\tilde P-P$ and observe that only $k_{00}$ is non-trivial. We thus obtain that $(P-\tilde P)\phi(H)$ and $\phi(H)(P-\tilde P)$  are trace-class. 

The previous two results show that $t_1:=\Tr [U(H),\tilde P] U^*(H)$ and $t_2:=\Tr [U(H),P] U^*(H)$ are defined and that $\Tr [U(H),\tilde P]=\Tr [U(H),P]=0$ so that $\Tr [W(H), \tilde P-P]=0$.  Also, $(\tilde P-P) W^*(H)$ is trace-class. Now, with $W=W(H)$, $t_1-t_2$ is given by 
\[
  \Tr [W,\tilde P]W^* - \Tr [W,P]W^*= \Tr  W(\tilde P-P)W^*- (\tilde P-P) WW^* = \Tr (\tilde P-P)(W^*W-WW^*)=0
\]
since $\Tr W(\tilde P-P)W^*=\Tr (\tilde P-P)W^*W$ as $W$ is bounded and $(\tilde P-P)W^*$ is trace-class and since $WW^*=W^*W$.  This shows that $\Tr [U(H),\tilde P] U^*(H)=\Tr [U(H),P] U^*(H)$.

That $T$ is then a Fredholm operator with index given by  $\Tr [U(H),\tilde P] U^*(H)$ for $\tilde P$ a projector is a non-trivial consequence of the trace-class nature of $[\tilde P,U(H)]$ and is shown in  \cite[Theorems 4.1\&5.3]{AVRON1994220}. This concludes the proof of the lemma.
\end{proof}
%The relation between $\sigma_I$ and a trace may also work for $\tilde P$ a projection but this is not needed.
We now relate the Fredholm operator $T$ and the calculation of its index as a trace with the line conductivity $\sigma_I=\sigma_I(H)$ defined in \eqref{eq:sigmaI}. We have the result:
%%%%%%%%%%
\begin{theorem}\label{thm:cond} Under the assumptions of Lemma \ref{lem:indexT}, we have:
 $2\pi \sigma_I  = {\rm Tr}[U(H),P] U^*(H) = {\rm Tr}[U(H),\tilde P] U^*(H) = {\rm Index}\ \tilde PU(H) \tilde P _{{\rm Ran} \tilde P}$.
\end{theorem}
%%%%%%%%%%
\begin{proof}
 This is essentially \cite[Proposition 4.3]{bal2022topological} with slightly different assumptions. The last two relations were proved in Lemma \ref{lem:indexT}. We focus on the first one. Let $g\in C^\infty_c(\Rm)$ and $\chi\in C^\infty_c(\Rm)$ with $\chi=1$ on the support of $W$ and $g$.  Let  $W_p$ be a sequence of polynomials chosen such that $\chi(W-W_p)$ and $\chi(W'-W'_p)$ converge to $0$ uniformly on $\Rm$ as $p\to\infty$. We find, with $\delta W_p:=W-W_p$
\[
  [\delta W_p, P] g\chi = \delta W_p[P,g]\chi + \delta W_p g[P,\chi] + [\delta W_p g\chi,P].
\]
We deduce from Lemma \ref{lem:indexT} and its proof that ${\rm Tr} [\delta W_p g\chi,P]=0$ and that $\delta W_p[P,g]$ is trace-class since $[P,g]$ has symbol in $S^0(\aver{X}^{-\infty})$. Since ${\rm Tr}AB={\rm Tr}BA$ when $A$ is trace-class and $B$ is bounded, we find that ${\rm Tr} \delta W_p[P,g]\chi ={\rm Tr} \chi\delta W_p[P,g]= {\rm Tr} [P,g] \chi\delta W_p$.  Therefore, 
\[
  {\rm Tr}  [W-W_p, P] g = {\rm Tr} [P,g](W-W_p)\chi + {\rm Tr} (W-W_p)g [P,\chi] \to 0 \quad \mbox{ as } p\to\infty.
\]
It remains to analyze ${\rm Tr}[W_p,P]g$. We verify that 
\[
{\rm Tr} [ H^n ,P] g =  {\rm Tr} \ n[ H ,P] H^{n-1} g ,\quad 
  {\rm Tr} [ W_p ,P] g =  {\rm Tr} [ H ,P] W_p' g.
\]
Indeed, from $[AB,C]=A[B,C] +[A,C]B$,
\[\begin{array}{l}
  {\rm Tr} [H^{n+1},P]g={\rm Tr} H^n[H,P]g+ [H^n,P]Hg = {\rm Tr} H^n[H,P]g\chi + [H^n,P]Hg \\
  = {\rm Tr} [H,P]H^ng\chi + [H^n,P]Hg = {\rm Tr} [H,P]H^ng+ [H^n,P]Hg = {\rm Tr}[H,P](n+1)H^ng
\end{array}\]
using that $H^n[H,P]g$ is trace-class so that $ {\rm Tr} H^n[H,P]g\chi= {\rm Tr} H^n\chi [H,P]g$ and for the last equality an induction in $n\geq1$. This proves the result for $W_p$ as well. Is remains to realize that $(W_p'-W')g$ is uniformly small as $p\to\infty$ to obtain that
\begin{equation}\label{eq:tracediff}
  {\rm Tr} [W,P] g = {\rm Tr} [H,P] W'g.
\end{equation}

We next compute
\[
  {\rm Tr} [U,P]U^*={\rm Tr} [W,P] + {\rm Tr}[W,P]W^* = {\rm Tr} [H,P] W' W^*={\rm Tr} [H,P] W' U^* -{\rm Tr} [H,P] W' .
\]
We now show that $0={\rm Tr} [H,P] W'$. Let $1=\psi_1^2+\psi_2^2$ a partition of unity with $0\leq \psi_j\leq 1$ for $j=1,2$ and such that $\psi_1\in C^\infty_c(\Rm)$ equals $1$ on the support of $W$. Then, using \eqref{eq:tracediff}, with $\psi_j=\psi_j(H)$, we have
$
  {\rm Tr} [H,P] W' = {\rm Tr} [H,P] W' \psi_1^2 = {\rm Tr} [W(H),P] \psi_1^2.
$
Now, since $[W(H),P]$ is trace-class,
\[
  {\rm Tr}[W,P] \psi_2^2= {\rm Tr}\  \psi_2 [W,P] \psi_2 =0
\]
since $W(H)\psi_2(H)=0$. Since ${\rm Tr}[W,P]=0$, we have ${\rm Tr} [H,P] W'={\rm Tr} [W,P] \psi_1^2={\rm Tr}[W,P]=0$.
This proves that ${\rm Tr} [U,P]U^* = 2\pi i {\rm Tr}[H,P]\varphi'(H)$ since $W'U^*=U'U^*=2\pi i\varphi'$.
\end{proof}
%%%%%%%%%%

\medskip

We next need to ensure that the index of the Toeplitz operator $T=\tilde PU(H) \tilde P _{{\rm Ran} \tilde P}$ is appropriately stable. We first need to show that the index can be computed by approximating it by an isotropic symbol.
%%%%%%%%%%
\begin{lemma}\label{lem:contiso}
 Let $T$ be as above with $H\in \ow ES^m_{d-1}(g^s)$. Then there is a sequence of operators $H_\eps$ for $0\leq\eps\leq1$ with symbol in $ES^m_{d-1}(g^{\rm i})$ for all $\eps>0$ and such that the corresponding $[0,1]\ni\eps\to T_\eps=\tilde P U(H_\eps) \tilde P_{{\rm Ran} \tilde P}$ is continuous in the uniform sense and $T_0=T$. Thus $\ind T_\eps$ is defined, independent of $\eps$ and equal to $\ind T$. Moreover, the symbols $a_\eps$ are chosen so that for any compact domain in $X=(x,\xi)$, $a_\eps=a_{d-1}$ on that domain for $\eps$ sufficiently small. 
\end{lemma}
%%%%%%%%%%
\begin{proof}
  The proof is based on  \cite[Lemma 19.3.3]{H-III-SP-94} showing that symbols in $S(M,g^s)$ may be suitably approximated by symbols in $S(M,g^i)$ as follows. Let $v(r):\Rm_+\to\Rm_+$ be a smooth non-increasing function such that $v(r)=1$ on $[0,1]$ and $v(r)=2/r$ on $[2,\infty)$. Let $a\in S(M,g^s)$. We then define the family of regularized symbols:
\[
   a_\eps(X) = a(v(\eps |X|)x,v(\eps |X|)\xi).
\]  
We observe that $a_\eps(X)=a(X)$ for $\eps|X|\leq1$ while $a_\eps(X)=a(\frac{X}{\eps|X|})$ homogeneous of degree $0$ for $\eps|X|\geq2$. Then \cite[Lemma 19.3.3]{H-III-SP-94} (where the metrics $g^s$ and $g^i$ are called $g$ and $G$, respectively) proves that $a_\eps(X)\in S(M,g^s)$ uniformly in $0\leq \eps\leq 1$ (i.e., every semi-norm defining the symbol space is bounded for $a_\eps$ uniformly in $\eps$). Moreover, $a_\eps(X)\in S(M,g^i)$ when $\eps>0$ with a bound that now depends on $\eps$. Since $a_\eps(X)=a(X)$ for $\eps|X|\leq1$, we obtain that $a_\eps$ converges to $a$ in $S(M,g^s)$ as $\eps\to0$. 

We now use the proof of \cite[Theorem 19.3.1']{H-III-SP-94} extending the index theorem \eqref{eq:FH} from the isotropic metric $g^i$ to the metric $g^s$. For $a_\eps(X) = a_{d-1}(v(\eps |X|)x,v(\eps |X|)\xi)$, we find that $a_\eps\in ES^m_{d-1}(g^s)$ uniformly in $\eps$ and $a_\eps\in ES^m_{d-1}(g^i)$ for $\eps>0$.  Let  $H_\eps=\ow a_\eps$ and $\tilde T_\eps=\tilde P U(H_\eps) \tilde P+I-\tilde P$. By uniformity of $a_\eps(X)\in S^m_{d-1}(g^s)$ in $\eps$ and uniformity of bounds in Lemma \ref{lem:fccalc} (see Remark \ref{rem:unifbd}), we obtain that $W(H_\eps)$ has symbol in $\tilde S^0(g^s)$ uniform in $\eps$. 

We now observe that $\tilde T^*_\eps \tilde T_\eps-I=\tilde P[W^*(H_\eps),\tilde P]U(H_\eps)\tilde P$ and $\tilde T_\eps \tilde T^*_\eps-I$ are uniformly compact in $\eps$ and even uniformly trace-class from the results of Lemma \ref{lem:indexT}. We then apply   \cite[Theorem 19.1.10]{H-III-SP-94} to obtain that the indices of $\tilde T_\eps$ and $\tilde T^*_\eps$, and hence that of $T_\eps$, are independent of $0<\eps\leq 1$. In the limit $\eps\to0$, this is the index of $T$. Thus $\ind T=\ind T_\eps$ for $\eps>0$ but now for a symbol $a_\eps\in ES^m_{d-1}(g^i)$.
\end{proof}

The above result shows that we can replace the symbol in $H_{d-1}$ by that obtained at $\eps>0$. We also observe that \eqref{eq:FH} is independent of $\eps$ for $\eps$ small. We may therefore assume that $a_{d-1}\in ES^m_{d-1}(g^i)$ in the computation of $\ind T$. The main advantage of the more constraining metric $g^i$ is that the corresponding symbol classes are now invariant under suitable rotations and permutation of the phase space variables $X$. The following result is then used. Let $Y=(x_1,\ldots,x_{d-1},\xi_1,\ldots \xi_{d-1})$.
%%%%%%%%%%
\begin{lemma}\label{lem:conttf}
  Let $g=g^i$. Let $[0,1]\ni t \to L_t$ be a continuous family of linear invertible transformations in $GL(2d-2,\Rm)$ in the $Y$ variables leaving the variables $(x_d,\xi_d)$ fixed.
  %, including dilations and rotations.  
  Let $a(X)\in ES^m_{d-1}(g^i)$. Then $a(t,X)=a(L_t X)\in ES^m_{d-1}(g^i)$. Let $T_t$ be the corresponding Toeplitz operator. Then $T_t$ is Fredholm with index independent of $t\in [0,1]$.
\end{lemma}
%%%%%%%%%%%%
%%%%%%%%%%%%
\begin{proof}
 %The proof has to be done in detail. Let $H_t=\ow a(t,X)$. 
We compute
\[
  U(H_t)-U(H_s)=W(H_t)-W(H_s)=-\frac1\pi \dint_{\Cm} \bar\partial \tilde W(z) (z-H_t)^{-1} (H_t-H_s)(z-H_s)^{-1} d^2z.
\]
We use $(z-H_s)^{-1}=(i-H_s)^{-1}(I+(i-z)(z-H_s)^{-1})$ and $\|(z-H_s)^{-1}\|_{{\cal L}(L^2)}\leq |{\rm Im} z|^{-1}$. We choose $|\bar\partial \tilde W(z)|\lesssim |{\rm Im} z|^{2}$. We know that $(i-H_s)^{-1}= \ow r_{s}$ with $r_{s}\in S(M_{d-1}^{-1},g^i)$ uniformly in $s\in[0,1]$ so that $(H_t-H_s) (i-H_s)^{-1}=\ow (a_t-a_s)\sharp r_s$.
Then, from \cite[(18)]{bony1994espaces} and composition calculus, there exists a seminorm $k$ independent of $t$ such that
\[
  \|(H_t-H_s) (i-H_s)^{-1}\| \lesssim \|a(t,X)-a(s,X)\|_{k;S(M_{d-1},g^i)} \|r_{s}\|_{k;S(M_{d-1}^{-1},g^i)} 
\]
which is bounded by $C\|a(t,X)-a(s,X)\|_{k;S(M_{d-1},g)}$. This involves contributions of the form of powers of $(L_t-L_s)X$ times derivatives of $a$ by chain rule.  We thus obtain terms of the form $X_i\partial_{X_j}a$ (with $X_j\not\in\{x_d,\xi_d\}$), which are operations that are stable from $S(M_{d-1},g^i)$ to itself provided that $g=g^i$. Note that the vector field (implementing rotation in the variables $(x_j,\xi_j)$) $\xi_j\partial_{x_j}-x_j\partial_{\xi_j}$ does not preserve $S(M_{d-1},g^s)$ and hence the importance of working with symbols in the smaller isotropic class.

Higher-order derivatives are bounded in the same way, allowing us to obtain that  $C\|a(t,X)-a(s,X)\|_{k;S(M_{d-1};g)}$ is bounded by a constant times $|t-s|$. Therefore, $U(H_t)-U(H_s)$ is small in the uniform sense for small $(t-s)$ so that the index of $T_t$ is continuous in $t$ and hence independent of $t\in[0,1]$.
\end{proof}
The above result states in particular that for $a_{d-1}\in ES_{d-1}(g^i)$, then the index of $T$ is independent of any rescaling $Y_j\to\lambda Y_j$ for $\lambda>0$ (leaving all other variables fixed) as well as any rotation in the phase space variables mapping $(Y_j,Y_k)$ to $(Y_k,-Y_j)$ (note the sign change to preserve orientation). We will use the above lemma only for such transformations (dilations and permutations).

We could show similarly that $\sigma_I$ is independent of changes in $\varphi$ although this property will automatically come from \eqref{eq:tcc} proved in the next section.  However, we need the following straightforward result.
%%%%%%%%%%
\begin{lemma}\label{lem:contind}
  Let $g\in \{g^s,g^i\}$. For $t\in [0,1]$, let $t\mapsto a_t\in ES_{d-1}^m(g)$ be a continuous path of elliptic symbols. Then the indices of the corresponding Fredholm operators $F_t=\ow (a_t-im_d)$ and $T_t=\tilde P U(\ow a_t)\tilde P_{|{\rm Ran}\tilde P}$ are independent of $t\in [0,1]$. 
\end{lemma}
%%%%%%%%%%
\begin{proof}
  By assumption and construction, $\Lambda^{-m}F_t$ and $T_t$ are continuous in time as operators from $\tilde H^0$ to  itself. Their indices are therefore constant in time.
\end{proof}
We apply the preceding lemma to $a_1=a(x_{d-1}',x_d,\xi)$ and $a_0=a(x_{d-1}',0,\xi)$ while $a_t=ta_1+(1-t)a_0$.  The path of symbols belongs to $ES_{d-1}^m$ so that the indices of the respective operators are defined with clear continuity in $t$. This allows us to replace an $x_d-$dependent elliptic symbol $a_{d-1}$ by an $x_d-$independent one, which is used below in the proof of the topological charge conservation.  
%%%%%%%%%%

%
%
\section{Topological Charge Conservation}
\label{sec:tcc}
This section shows that the two classifications in Theorems \ref{thm:FH} and \ref{thm:cond} are in fact the same since we have the following topological charge conservation:
%%%%%%%%%%%%%%%%
\begin{theorem} \label{thm:tcc}
  Let $a\in E\tilde S^m$. Then we have 
\begin{equation}\label{eq:tcc}
  2\pi \sigma_I  =  - \dfrac{(d-1)!}{(2\pi i)^d (2d-1)!} \dint_{\Sm_R^{2d-1}} {\rm tr} (a^{-1}da)^{2d-1} = {\rm Index}\  F.
\end{equation}
\end{theorem}
%%%%%%%%%%%%%%%
%Above, $\Sm_R^{2d-1}$ is the sphere of radius $R$ in $\Rm^{2d}$ for $R$ sufficiently large that $a^{-1}$ is defined outside of the unit ball of radius $R$ in the variables $(x,\xi)$.  
%
%The derivation is done as in GB3 and QB with the caveat that we are in higher dimension and we need the symbol in $(y,\zeta)$ to be entirely isotropic. This is related to the conservation in Volovik and Essin-Gurarie. Such a conservation may not hold for bounded domain walls, something we avoid by considering isotropic symbols. 
%
%%%%%%%%%%%%%%%
\begin{proof}
%%%%%%%%%%%%%%%%
The first step is to continuously deform $a_{d-1}(x_{d-1}',x_d,\xi)$ to $a_{d-1}(x_{d-1}',0,\xi)$ using Lemma \ref{lem:contind} and the paragraph that follows it. Note that all terms in \eqref{eq:tcc} are stable under this change of symbols.  We next use the approximation of a symbol in $a_{d-1}\in ES_{d-1}^m =ES_{d-1}^m(g^s)$ by $a_{d-1}\in ES_{d-1}^m(g^i)$ using Lemma \ref{lem:contiso}. Note that again, all terms in \eqref{eq:tcc} are stable under this change of symbols since both symbols agree on the support of $\Sm_R^{2d-1}$. To simplify, we change notation to $(y,x)=(x_{d-1}',x_d)$ and to $(\zeta,\xi)=\xi$ with the new $\zeta ,y \in \Rm^{d-1}$ and $\xi ,x\in\Rm$.  The symbol $a_{d-1}=a_{d-1}(y,\zeta,\xi)$. 

We finally replace $a(X)$ by $a_{d-1}(y,\zeta,\xi)-ix$, which has no effect on the definition of $\sigma_I$. We observe that the integral in \eqref{eq:FH} recalled in \eqref{eq:tcc} is also not modified. The reason is that $a$ maps a sphere to $GL(n_{d-1};\Cm)$ and that the winding number of such functions classifies them and is certainly invariant by continuous deformation \cite{bott1978some}. Such a continuous deformation is based on the path $tx_d+(1-t)\aver{x_d}^{m-1}x_d$ with corresponding mapping $a_t$ in $GL(n_{d-1};\Cm)$ for all $t$ and hence continuity of the winding number. For the same reason, we may replace $a$ by $a_{d-1}(y,\zeta,\xi)+\alpha-ix$ for any $\alpha\in\Rm$ by continuity in $\alpha$.

Since $a_{d-1}$ does not depend on $x$, we introduce the partial spectral decomposition $\hat H=\hat H[\xi]$ such that 
\[
  H = \mF_{\xi\to x}^{-1} \dint_{\Rm}^{\oplus} \hat H[\xi] d\xi \ \mF_{x\to\xi},
\qquad \hat H[\xi]=\ow a_{d-1}(Y,\xi)
\]
with $Y=(y,\zeta)$ and a Weyl quantization in the variables $Y$ for each parameter $\xi\in\Rm$.
%Since $a_{d-1}$ is independent of $x$, we find that $\hat H=\ow a_{d-1}$. 
We thus obtain, using the trace-class properties of Lemma \ref{lem:indexT} providing traces as integrals of Schwartz kernels along diagonals, that
\[\begin{array}{rcl}
 2\pi\sigma_I &=&2\pi  {\rm Tr}_y \dint_{\Rm} i[H,P] (x,x') \varphi'(H)(x',x) dx'dx \\
  &=&  2\pi i {\rm Tr}_y \dint_{\Rm}  - x'H(x') \varphi'(H)(x'-x) dx' = {\rm Tr}_y \dint_{\Rm} \partial_\xi \hat H \varphi'(\hat H[\xi]) d\xi.
\end{array}\]
Here, ${\rm Tr}_y$ denotes the integration in all variables but $x=x_d$ (using Fubini), $H(x-x')$ is the dependence in $(x,x')$ of the Schwartz kernel of the operator $H$ and we used 
\[
  \phi(x') := \dint_{\Rm} (P(x-x')-P(x))dx =-x' \qquad \mbox{ since } \qquad \partial_{x'} \phi(x') =-\dint_{\Rm}  P'(x-x')dx=-1.
\]
Following \cite{bal2022topological}, we introduce the complex variable  $\Cm\ni z=\lambda+i\omega$ and identify the spatial variable $x=x_d$ with the imaginary part $\omega$. The dual variable $\xi$ is considered as another parameter and pseudo-differential operators and semiclassical operators are now defined in the variables $Y=(y,\zeta)$. We denote by $\sigma_z(y,\zeta,\xi)=z-a_{d-1}(y,\zeta,\xi)=\lambda-a(y,\omega,\zeta,\xi)$ the symbol of $z-\hat H[\xi]$. 

We now introduce the semiclassical parameter $0 <h\leq 1$ and the operator $\hat H_h$ with symbol $a_{d-1}(y,h\zeta,\xi)$. Using semiclassical notation (in the phase-space variable $Y$), we thus observe that $z-\hat H_h= \ow_h \sigma_z$. Using Lemma \ref{lem:resiso}, we know that $ (z-\hat H_h)^{-1}=\ow_h r_z$ is a PDO with semiclassical symbol $r_z(y,\zeta,\xi;h)$. Note that the latter term has a complicated dependence on $h$ that  will be made explicit asymptotically using that $r_z\sharp_h\sigma_z=I$. We know from Lemma \ref{lem:conttf} that $\sigma_I(H_h)$ is independent of $0 < h\leq 1$. We may therefore compute it for  $0<h\leq h_0$ with $h_0$ sufficiently small based on terms in the above integral that are independent of $h$ as $h\to0$.

We know from Lemma \ref{lem:fccalc} (or from the semiclassical version of the functional calculus \cite[Theorem 8.7]{dimassi1999spectral} for $h_0$ small enough) that $\varphi'(\hat H)$ is a PDO and define $s$ such that $\partial_\xi \hat H \varphi'(\hat H)=\ow s$. We therefore obtain from Lemma \ref{lem:indexT} that
\[
  2\pi\sigma_I = {\rm Tr}_{y} \dint_{\Rm} \partial_\xi \hat H \varphi'(\hat H) d\xi =\dfrac{1}{(2\pi)^{d-1}} \dint_{\Rm^{2d-1}}  {\rm tr}\  s(Y,\xi) dY d\xi.
\]
Passing to the semiclassical regime $\zeta\to h\zeta$ with $\partial_\xi \hat H_h \varphi'(\hat H_h)=\ow_h s$ with now $s(y,\zeta,\xi;h)$, we have
\[
   2\pi\sigma_I = {\rm Tr}_{y} \dint_{\Rm} \partial_\xi \hat H_h \varphi'(\hat H_h) d\xi = \frac{1}{(2\pi h)^{d-1}}\dint_{\Rm^{2d-1}} {\rm tr}\  s(Y,\xi;h) dYd\xi.
\]
However, we have the following bound coming from the functional calculus Lemma \ref{lem:fccalc} (or \cite[Theorem 8.7]{dimassi1999spectral} for $h_0$ small enough)
\[
  |\partial^\beta s(Y,\xi;h)| \leq C_{N,\beta} \aver{Y,\xi}^{-N}
\]
uniformly in $0<h\leq h_0$. Moreover, we have from semiclassical calculus \cite{dimassi1999spectral,zworski2012semiclassical} the decomposition
\[
  s(Y,\xi;h) = \dsum_{j=0}^M h^j s_j(Y,\xi) + h^{M+1} \rho_M(Y,\xi;h)
\]
with $\rho_M$ bounded in $S^0(\aver{Y,\xi}^{-N})$ for any $N$ (uniformly in $0<h\leq h_0$). We thus observe that
\[
    2\pi\sigma_I = \frac{1}{(2\pi )^{d-1}}\dint_{\Rm^{2d-1}} {\rm tr}\   s_{d-1}(Y,\xi) dY d\xi = \lim_{R\to\infty}  \frac{1}{(2\pi )^{d-1}}\dint_{[-R,R]^{2d-1}}  {\rm tr}\  s_{d-1}(Y,\xi) dYd\xi.
\]
It thus remains to identify the term of order $d-1$ in the expansion in powers of $h$ of $s(h)$. Let $\varsigma$ be such that $\varphi'(\hat H_h)=\ow_h\varsigma$ while $\partial_\xi \hat H_h=-\partial_\xi \sigma_z$. Then 
\[
  s=-\partial_\xi \sigma_z \sharp_h \varsigma,\qquad \varsigma(Y,\xi;h) = -\frac1\pi \dint_{\Cm} \bar\partial \tilde\varphi' r_z(Y,\xi;h) d^2z
\]
with as we recall $(z-\hat H_h)^{-1} = \ow_h r_z$ using the Helffer-Sj\tio strand formula \eqref{eq:hs}.

We can replace the integral over $\Cm$ by an integral over a compact domain $Z$ since $\tilde\varphi'$ may be chosen compactly supported. We then define $Z_\delta=Z\cap \{|{\rm Im}z|\geq\delta\}$ and $z=\lambda+i\omega$. Since $r_z(Y;h)$ is controlled as $\omega\to0$ uniformly in $[-R,R]^{2d-1}$ by \eqref{eq:resbd} and $\tilde\varphi'$ is of order $\omega^N$ for any $N$ for $|\omega|$ small, we can choose $\delta=h^\gamma$ for any $\gamma>0$ and observe that up to an error of order $h^N$ for any $N$, we can replace $Z$ by $Z_\delta$ in the calculation of the trace.

The symbol $r_z$ also admits the expansion
\[
  r_z(Y,\xi;h) = \dsum_{j=0}^N r_{jz}(Y,\xi) h^j + \rho_{z}(Y,\xi;h) h^{N+1}
\]
with $h^{N+1} \rho_z$ uniformly bounded by $h^{d+1}$, say, on $[-R,R]^{2d-1}\times Z_\delta$ for $h$ small.

In the integral of $s_{d-1}$, it is therefore sufficient to look at the terms of order $d-1$ 
\[
  (-\partial_\xi\sigma_z \sharp_h r_z)_{d-1} = \dsum _{j=0}^{d-1} (-\partial_\xi\sigma_z \sharp_h r_{jz})_{d-1-j}.
\]
We recall that the Moyal product $\sharp_h$ in \eqref{eq:sharph} is in the variables $Y=(y,\zeta)$ with $z$ and $\xi$ as parameters.

The next step is to write $r_{jz}(Y,\xi)$ for each $j\geq0$ in terms of $\sigma$ and $\sigma^{-1}$ as well as their derivatives. %It is then that we can rescale each component in $(Y,\xi)$ and show that exactly one derivative in each variable contributes. 
Note that $\sigma_z \sharp_h r_z=r_z \sharp_h \sigma_z=I$. We have in generic variables 
\[
  a\sharp_h b (x,\xi) =  \Big( e^{i\frac h2(\partial_x\cdot\partial_\zeta - \partial_y\cdot\partial_\xi)} a(x,\xi) b(y,\zeta)\Big)_{|y=x;\zeta=\xi}  = : \dsum_{j\geq0} \frac{1}{j!} (\frac{-ih}2)^j \{ a,b\}^j (x,\xi).
\]
We defined
\[
   \{a,b\}^j(x,\xi) := \Big( (\partial_\xi \cdot \partial_y - \partial_x \cdot \partial_\zeta)^j  a(x,\xi) b(y,\zeta)\Big)_{|y=x;\zeta=\xi}.
\]
For $j=1$, this is $\{a,b\}$ the standard Poisson bracket. Thus,
\[
  I = r_z \sharp_h \sigma_z = \dsum_{j\geq0} h^j  r_{zj}  \sharp_h\sigma_z 
   = \dsum_{j,k} h^{j+k} \dfrac{1}{k!} \big( \frac{-i}2\big)^k  \{r_{zj}, \sigma_z\}^k.
\]
We thus obtain from $\sigma_z \sharp_h r_z=r_z \sharp_h \sigma_z=I$ the equations
\[
  \dsum_{j+k=l} \frac{1}{k!} \big(\frac{- i}2\big)^k \{r_{zj},\sigma_z\}^k = \delta_{l0},\qquad l\geq0.
\]
The leading equation is $r_{z0}=\sigma_z^{-1}$, which is defined for $\omega\not=0$. Then higher-order equation can be solved iteratively for $r_{zj}$. The next two equations are for instance
\[
  -\frac i2 \{r_{z0},\sigma_z\} +  r_{z1}\sigma_z=0,\qquad
   -\frac18 \{r_{z0},\sigma_z\}^2- \frac i2 \{ r_{z1},\sigma_z\} +r_{z2} \sigma_z =0.
\]
We thus observe iteratively that $r_{zj}(Y,\xi)$ is a product of a maximum of $2j+1$ terms alternating a derivative (possibly of order $0$) of $\sigma_z$ with one  (possibly of order $0$) of $\sigma_z^{-1}$. 

The same property holds for 
\begin{align}\label{eq:rzjk}
\partial_\xi \sigma_z \sharp_h r_z = \dsum_{j,k} h^{j+k} \dfrac{1}{k!} \big( \frac{-i}2\big)^k \{ \partial_\xi \sigma_z,r_{zj}\}^k,
\end{align}
and more specifically for the terms of interest $j+k=d-1$ in the calculation of the trace. 

Let us now rescale $\tau\to \eta \tau$ for $\tau$ one variable in $(Y,\xi)\to (Y_\eta,\xi_\eta)$. By lemma \ref{lem:conttf}, the trace is independent of $\eta$ when $\hat H$ now has symbol $a_{d-1}(Y_\eta,\xi_\eta)$. Let $s_{d-1}(Y,\xi;\eta)$ be the corresponding symbol appearing in the trace calculation. We thus have
\[
  \dint  {\rm tr}\   s_{d-1}(Y,\xi;\eta) dY d\xi = \eta  \dint   {\rm tr}\   s_{d-1}(Y,\xi;1) dY d\xi .
\]
From the above considerations, we obtain
\[
  s_{d-1}(Y,\xi;\eta) = \prod_{j=1}^J \partial^{\alpha_j} \sigma^{\beta_j} (Y_\eta,\xi_\eta) = \eta^\gamma \prod_{j=1}^J (\partial^{\alpha_j} \sigma^{\beta_j})(Y_\eta,\xi_\eta)
\]
where $\beta_j\in \{-1,1\}$ and $\gamma$ is the number of derivatives in the variable $\tau$ that appear in $s_{d-1}$. Integrating the latter expression over $\Rm^{2d-1}$ and changing variables $(Y_\eta,\xi_\eta)\to (Y,\xi)$ shows that necessarily $\gamma=1$ in order for $2\pi \sigma_I$ to be independent of $\eta$. 

This implies that exactly one differentiation in each of the variables $(Y,\xi)$ appears in the terms that contribute to the integral of $s_{d-1}$.  Therefore, $\partial_\xi\sigma_z$ is the only term involving a derivative in $\xi$ and $r_{jz}$ contributes exactly one derivative in each of the components of $Y$. This implies that $j=d-1$ and $k=0$ is the only contributing term to the trace in \eqref{eq:rzjk}. The contributing terms to the trace are therefore a subset of $h^{d-1}  \partial_\xi \sigma_z r_{z,d-1}$. 

Since $\{\cdot,\cdot\}^j$ applies $j$ derivatives,  only $j=1$ contributes to the trace integral.  We thus obtain
\[
  r_{zj} = \frac i2 \{ r_{z,j-1},\sigma_z\} \sigma_z^{-1} \ +\  {\rm non\ contributing\ terms\ to\ the\ trace\ integral}.
\]
Let $\tilde r_{zk}$ collect the above contributing terms. We obtain 
\[
  \tilde r_{zj} =  \frac i2 \{ \tilde r_{z,j-1},\sigma_z\} \sigma_z^{-1} = (\frac i2)^2 \{  \{ \tilde r_{z,j-2},\sigma_z\} \sigma_z^{-1},\sigma_z\} \sigma_z^{-1} = \big(\frac i2)^j \{\sigma_z^{-1},\sigma_z\}^j_f \sigma_z^{-1}
\]  
where $\{\sigma_z^{-1},\sigma_z\}^j_f$ is the subset of $\{\sigma_z^{-1},\sigma_z\}^j$ where differentiation in a pair $(y_k,\zeta_k)$ appears at most once.  Denote $c_j=(\frac i2)^j$.  We have for $k=d-1$ the expression
\[
   \{ \sigma_z^{-1},\sigma_z\} ^{d-1}_f  = \dsum_{\rho \in \mS_{d-1}}  \{ \sigma_z^{-1},\sigma_z\}_{\rho_1} \ldots \{ \sigma_z^{-1},\sigma_z\}_{\rho_{d-1}}
\]
where the sum is over $\rho\in \mS_{d-1}$ the set of permutations of $\{1,\ldots,d-1\}$. We thus obtain with $\R_h=[-h^{-1},h^{-1}]^{2d-1}$
\[
  2\pi \sigma_I = \lim_{h\to0} \frac{2c_{d-1}}{(2\pi)^d} \dint_{\R_h\times Z_\delta}  \bar\partial \tilde\varphi'(z) \  {\rm tr} \ \sigma_z^{-1}\partial_\xi \sigma_z \{ \sigma_z^{-1},\sigma_z\} ^{d-1}_f  d^2z dY d\xi
\]
by cyclicity of the trace. Note that $\delta=\delta(h)$. We denote $Z_{\delta\pm}=Z_\delta\cap \{\pm\omega>0\}$. Applying the Stokes theorem on each connected component of $Z_\delta$ realizing that $\sigma_z$ and $\sigma_z^{-1}$ are analytic in $z$, we find
\[
  2\pi \sigma_I = \lim_{h\to0} \frac{-i c_{d-1}}{(2\pi)^d} \dint_{\R_h\times \partial Z_{\delta\pm}}  \tilde{\varphi}'(z) \  {\rm tr} \ \sigma_z^{-1}\partial_\xi \sigma_z \{ \sigma_z^{-1},\sigma_z\} ^{d-1}_f  dzdY d\xi
\]
where $z=\lambda \pm \delta$ on the two disconnected components of $\partial Z_\delta$ where $\tilde\varphi'(z)$ does not vanish. Also, $dz$ is oriented in the counterclockwise direction. Note that $\tilde{\varphi'}(\lambda+i\delta)= \varphi'(\lambda) + O(h^N)$ for any $N>0$ by construction of the almost analytic extension for $\delta=h^\gamma$ for any $\gamma>0$. We will obtain that the above integrand is independent of $\lambda$ and recall that $\varphi'(\lambda)$ integrates to $1$. We recast the above as 
\[
   2\pi \sigma_I = \lim_{h\to0}   \dint_{\Rm} \varphi'(\lambda)  \sum_{\eps=\pm1} \eps \Phi_h(\lambda+ i\eps\delta) d\lambda;\
     \Phi_h(z)  =  \frac{-i c_{d-1}}{(2\pi)^d} \dint_{\R_h} {\rm tr} \ \sigma_z^{-1}\partial_\xi \sigma_z \{ \sigma_z^{-1},\sigma_z\} ^{d-1}_f  dY d\xi.
\]
The term $\eps$ in front of $\Phi_h$ indicates the orientation of the integral along the boundary.  We are therefore interested in computing the limit as $h\to0$ of $\Phi_h(\lambda+i\delta) -\Phi_h(\lambda-i\delta)$. 

The computation of the invariant involves that of $(\sigma_z^{-1}d\sigma_z)^{2d-1}$. For $\sigma=\sigma(x)$ in $n$ dimensions, we have
\[
   (\sigma^{-1}d\sigma)^n = \dsum_{\rho \in \mS_n}   (-1)^\rho\ \sigma^{-1} \partial_{\rho_1}\sigma \ldots  \sigma^{-1} \partial_{\rho_n}\sigma\  dx_1 \wedge \ldots \wedge dx_n,
\]
where  $(-1)^\rho=\epsilon_{\rho_1,\ldots,\rho_n}$ is the signature of the permutation $\rho: (1,\ldots,n)\to (\rho_1,\ldots \rho_n)$. By cyclicity of the trace, the term $\sigma_z^{-1}\partial_\xi\sigma_z$ can always be brought to the left of the product. However, $ \{ \sigma_z^{-1},\sigma_z\} ^{d-1}_f$ involves a summation over only specific permutations of the variables $Y$. It is where having a symbol in an isotropic class with $g=g^i$ is used. We have seen in Lemma \ref{lem:conttf} that any rotation in the variables $Y$ did not change $\sigma_I$. Therefore, any permutation of the variables in $Y$ with positive determinant leads to the same $\sigma_I$ and any permutation with negative determinant leads to $-\sigma_I$. 

Note that 
\[
   \{ \sigma_z^{-1},\sigma_z\} _j= -\sigma_z^{-1} \partial_{\zeta_j} \sigma_z \sigma_z^{-1} \partial_{y_j} \sigma_z +\sigma_z^{-1} \partial_{y_j} \sigma_z \sigma_z^{-1} \partial_{\zeta_j} \sigma_z .
\]
So, all terms of the form $\sigma_z^{-1} \partial_{y_j} \sigma_z \sigma_z^{-1} \partial_{\zeta_j} \sigma_z$ come with positive orientation while the terms with $(y_j,\zeta_j)$ permuted come with negative orientation.

Let $\sigma_I(Y)$ be computed as above. We then find that 
\[
  \sigma_I = \frac{1}{(2d-2)!} \dsum_{\rho\in\mS_{2d-2}} (-1)^\rho \sigma_I(\rho (Y))
\]
where summation is over all permutations of ${1,\ldots,2d-2}$ and $(-1)^\rho$ is the signature of the permutation. Combining the permutations generating $\{\cdot,\cdot\}^{d-1}_f$, we observe that each term $\prod_{j=1}^{2d-2}\sigma_z^{-1}\partial_{\rho_j(Y)}\sigma_z  (-1)^\rho$ appears $\gamma_d=2^{d-1}(d-1)!$ times, where $2^{d-1}$ comes from the difference of products in each Poisson bracket and $(d-1)!$ from the possible permutations of the variables.  Therefore we have $2\pi\sigma_I$ as above with now, using $c_{d-1}2^{d-1} = i^{d-1}$,
\[
   \Phi_h(z) %= \dfrac{-ic_{d-1}}{(2\pi)^d}\dfrac{2^{d-1}(d-1)!}{(2d-2)!} \frac{1}{2d-1} \dint_{\R_h}  {\rm tr}\ (\sigma_z^{-1}d \sigma_z)^{2d-1}
    = \dfrac{-i^{d}}{(2\pi)^d}\dfrac{(d-1)!}{(2d-1)!} \dint_{\R_h}  {\rm tr}\ (\sigma_z^{-1}d \sigma_z)^{2d-1}.
\]
The term $(2d-1)^{-1}$ comes from the fact that we placed the variables $\xi$ first.

We recall that $z=\lambda +i\delta$ for $\delta>0$. Let us consider the hypersurfaces $\{\pm \delta\}\times \R_h$ in the variables $(\omega,\xi,Y)$. When $h$ is sufficiently small, $\sigma^{-1}$ is defined on the hypersurface $C_h:=(-\delta,\delta)\times\partial \R_h$. Moreover, $(\sigma_z^{-1}d \sigma_z)^{2d-1}$ is integrable over that cylinder with negligible integral as $h\to0$. The reason is that $d\sigma_z$ is small there by assumption of $\sigma_z$ in an isotropic class. 
%%%%%
%%%%% WE MAY WANT MORE DETAIL HERE AND AT LEAST CHECK VERY CAREFULLY
%%%%%
Let us define $S_h=\cup \{\pm \delta\}\times \R_h \cup C_h$ a closed hypersurface. Then, to leading order, we have
\[ 
  \lim_{h\to0} \Big(  \Phi_h(\lambda+i\delta) -  \Phi_h(\lambda-i\delta) \Big) =\lim_{h\to0} \dfrac{(-1)^{d-1}}{(2\pi i)^d}\dfrac{(d-1)!}{(2d-1)!}  \dint_{S_h}  {\rm tr}\ (\sigma_z^{-1}d \sigma_z)^{2d-1},
\]
where the right-hand side is in fact independent of $h$ as we now prove. Indeed, we find that $d{\rm tr}\ (\sigma_z^{-1}d \sigma_z)^{2d-1}=0$ as a $2d-$form so that by the Stokes theorem, the above integral remains unchanged if $S_h$ is continuously deformed, for instance, to the sphere $\Sm_R$ of radius $R$ in the variables $(\omega,\xi,Y)$ so long as $\sigma^{-1}_z$ is defined. We then realize that $z=\lambda+i\omega$ and that the above integral is also independent of $t\lambda$ for $t\in[0,1]$ by continuity and integrality of the winding number. Since $\sigma_{i\omega}(Y,\xi)=-a(\omega,\xi,Y)$, this concludes the proof of the result modulo a sign. Upon inspection, we observe that the above integral has been computed for the orientation $dx_d \wedge d\xi_d \wedge dx_1\wedge d\xi_1 \ldots dx_{d-1}\wedge d\xi_{d-1}>0$ with $\omega\equiv x_d$ the first variables from the orientation $\Phi_h(\lambda+i\delta) -  \Phi_h(\lambda-i\delta)$. The latter orientation is $(-1)^{d} d\xi_1\wedge dx_1 \wedge \ldots d\xi_d \wedge dx_d$. With the latter choice of orientation, we obtain the topological charge conservation between the topological charge given by the index of $F$ and the transport asymmetry given by the conductivity $2\pi\sigma_I$.
%%%%%%%%%%%%%%%
\end{proof}
%%%%%%%%%%%%%%%

The above topological charge conservation generalizes to arbitrary dimension the bulk-interface correspondence that applies in dimension $d=2$. It is shown in \cite{bal2022topological} that the Fedosov-\horm formula may be interpreted as a bulk-difference invariant. Indeed, the three-dimensional sphere in \eqref{eq:tcc} may be continuously deformed to the union of two hyperplanes evaluated at $x_1\equiv y=\pm y_0$, corresponding to two $y-$independent (bulk) invariants. The line conductivity is thus given as a difference of bulk-invariants; see  \cite{bal2022topological}  for a more detailed presentation.

%
%%%
\section{Generalized Dirac operators and degree theory}
\label{sec:deg}
%%%
%

%For operators with the Clifford structure given in \eqref{eq:cliffordrep} below, the formula \eqref{eq:FH} significantly simplifies as the computation of the degree of a map related to the Hamiltonian $H_k$.

%\gb{From INTRODUCTION. Dirac operator no longer in introduction. Remove example most likely?}

%The ellipticity assumption on $a_k$ implies that $|\rh^k(X)|>0$ outside of a compact set. Let $B_R$ be the ball of radius $R$ in $\Rm^{d+k}$ in the variables $(x'_k,\xi)$ with $R$ sufficiently large that $|\rh^k(X)|>0$ on $\Rm^{d+k}\backslash B_R$ with $x_k''=0$. We can then define the degree of the map $\rh^k$. For $y_0$ a regular value of $\rh^k$ in the vicinity of $0$ (see section \ref{sec:deg}), then the degree of $\rh^k$ is defined as 
%\begin{equation}\label{eq:degexplicit}
%   \deg \rh^k  = \dsum_{\zeta \in (\rh^k)^{-1}(y_0)} \sign \det J(\zeta)
%\end{equation}
%with $J$ the non-degenerate Jacobian matrix of the map $\rh^k$.  In many applications, we may choose $y_0=0$. Several useful equivalent formulations of the degree are given in section \ref{sec:deg}. 
%
%The main result of section \ref{sec:deg} is Theorem \ref{thm:tccp}, stating that the topological charge of $H_k$ is, up to a sign depending on choices of orientation for the Clifford matrices and the phase space variables $X$, equal to the degree of $\rh^k$.

%
%%
%\paragraph{Clifford matrices and Dirac operators.}
%%
%
We assume in this section that  $a_k\in ES^m_k$ has the following form
\begin{equation}\label{eq:cliffordrep}
  a_k (X)= \rh^k(X) \cdot \Gamma_k
\end{equation}
where $\Gamma_k$ is a collection of matrices  in the representation of the Clifford algebra ${\rm Cl}_{n_k}(\Cm)$  that may be defined as follows. For $0\leq k\leq d$, define $\kappa:=\kappa_k=\lfloor \frac{d+k}2 \rfloor$ and $n_k=2^{\kappa_k}$. We assume that the matrices $\Gamma_k=(\gamma_\kappa^j)_j$ for $1\leq j\leq d+k$ satisfy the  commutation relations
\begin{equation}\label{eq:clifcomm}
  \gamma_\kappa^i \gamma_\kappa^j +  \gamma_\kappa^j \gamma_\kappa^i = 2\delta_{ij} I_{n_k}.
\end{equation}
These properties imply that $a_k^2=|\rh^k|^2I_{n_k}$ is proportional to identity.

Specifically, the matrices $\gamma_\kappa^j$ of level $\kappa$ are constructed starting from $\gamma^{1,2,3}_1=\sigma_{1,2,3}$ the standard Pauli matrices and then iteratively  as
\[
   \gamma_{\kappa+1}^j = \sigma_1 \otimes  \gamma^j_\kappa,\ \ 1\leq j\leq 2\kappa+1,\quad \gamma_{\kappa+1}^{2\kappa+2} = \sigma_2 \otimes I_{n_k} , \qquad  \gamma_{\kappa+1}^{2\kappa+3} = \sigma_3 \otimes I_{n_k} .
\]
The last matrix plays the role of the chiral symmetry matrix in even dimension $d+k=2\kappa_k$. The construction of the above matrices mimics the construction of the augmented Hamiltonians $H_j$ for $k<j\leq d$. When $d+k$ is even, the chiral symmetry is given by
\[
   \gamma_\kappa^0 a_k + a_k \gamma_\kappa^0 =0 ,\qquad \gamma_\kappa^0 := \gamma_\kappa^{2\kappa+1}.
\]
%Equivalently, this means that $\rh^\kappa_{2\kappa+1}=0$ when $d+k$ is even. 
 For $a_k=\rh^k\cdot \Gamma_k$ as above, we denote by $\rh^j$ for $k\leq j\leq d$ the vector fields of dimension $d+j$ such that the augmented Hamiltonians constructed in section \ref{sec:local} satisfy $H_j=\ow a_j$ with, as we verify, $a_j=\rh^j\cdot\Gamma_j$.
 
 Dirac operators are in the form \eqref{eq:cliffordrep}. In two dimensions, we have explicitly $\Gamma_0=(\sigma_1,\sigma_2)$ while $\rh^0(X)=(\xi_1,\xi_2)$ and $\Gamma_1=(\sigma_1,\sigma_2,\sigma_3)$ while $\rh^1(X)=(\xi_1,\xi_2,x_1)$. In dimension $d=3$, we have  $\Gamma_0=(\sigma_1,\sigma_2,\sigma_3)$ while $\rh^0(X)=(\xi_1,\xi_2,\xi_3)$, next $\Gamma_1=(\sigma_1\otimes\sigma_1,\sigma_1\otimes\sigma_2,\sigma_1\otimes\sigma_3,\sigma_2\otimes I_2)$ while  $\rh^1(X)=(\xi_1,\xi_2,\xi_3,x_1)$, and finally $\Gamma_2=(\sigma_1\otimes\sigma_1,\sigma_1\otimes\sigma_2,\sigma_1\otimes\sigma_3,\sigma_2\otimes I_2,\sigma_3\otimes I_2)$ while  $\rh^2(X)=(\xi_1,\xi_2,\xi_3,x_1,x_2)$. When $d=3$, then $\kappa_0=1$ while $\kappa_1=\kappa_2=2$ for a maximum of matrices satisfying \eqref{eq:clifcomm} equal to $2\kappa_2+1=5$. Several other examples will be presented in the next section.

\paragraph{Topological charge computation}
For elliptic operators that admit the above Clifford representation \eqref{eq:cliffordrep}, the computation of the index in \eqref{eq:FH} significantly simplifies as the computation of the degree of the map $\rh^k$. 

We recall the definition of the degree of a map following \cite[Chapters 1.3\&1.4]{nirenberg1974topics}; see also \cite[Chapters 13\&14]{DFN-SP-1985}. Let $C$ be an open set in $\Rm^n$ with compact closure $\bar C=C\cup \partial C$. Let $\rh:\bar C\to\Rm^n$ be a sufficiently smooth map such that $|\rh(\zeta)|>0$ for $\zeta\in \partial C$.  There are regular values $y_0$ of $\rh$ arbitrarily close to $0$ by Sard's theorem that allow us to define the degree of $\rh$ as 
\begin{equation}\label{eq:indexsum}
  \deg (\rh,\bar C,0) = \dsum_{\zeta\in \rh^{-1}(y_0)} {\rm sgn} \det  J_{\rh}(\zeta).
\end{equation}
The above sum ranges over a finite set and is independent of the regular value $y_0$ in an open vicinity of $0$.

The definition of the index of a map from $M$ to $N$ depends on the chosen orientation on $M$. We consider two natural orientation in the context of topological insulators.  Let $B_d\subset\Rm^{2d}$ the ball of radius $R$ given by $\{|X|\leq R\}$. By ellipticity assumption, $|\rh^d|>0$ on $\partial B_d$ for $R$ large enough. We now define degrees for $\rh^d$ with two possible orientations:
\begin{equation}\label{eq:degrees}
\begin{array}{rcl}
   \tdeg \rh^d &: =& \deg (\rh^d, \bar B_d,0) \ \mbox{ with $\bar B_d$ oriented as } d\xi_1 dx_1\ldots d\xi_d dx_d>0, \\
    \deg \rh^d &: =& \deg (\rh^d, \bar B_d,0) \ \mbox{ with $\bar B_d$ oriented as } d\xi_1\ldots d\xi_d dx_1 \ldots dx_d>0.
\end{array} 
\end{equation}
We observe that 
\begin{equation}\label{eq:changeorient}
    \deg \rh^d = (-1)^{\frac12 d(d-1)} \tdeg \rh^d.
\end{equation}
The degree $\tdeg$ is naturally related to $\ind F$ while the degree $\deg$ is more naturally related to that of $\rh^k$ as we now describe.

Using Lemma \ref{lem:contind}, we obtain that the index is unchanged if $\rh^k(x_k',x_k'',\xi)$ is replaced by $\rh^k(x_k',\xi):=\rh^k(x_k',0,\xi)$ in the definition of the symbol. We may therefore see $\rh^k$ as a map from $\Rm^{d+k}$ to $\Rm^{d+k}$ such that, by ellipticity, $|\rh^k|\geq h_0>0$ for $|(x_k',\xi)|\geq R$.  Let $B_k=\{|(x_k',\xi)|<R\}$. Then we define 
\begin{equation}\label{eq:degk}
  \deg \rh^k := \deg (\rh^k, \bar B_k,0) \ \mbox{ with $\bar B_d\subset\Rm^{2d}$ oriented as } d\xi_1\ldots d\xi_d dx_1 \ldots dx_k>0.
\end{equation}
The orientation of $B_k$ is inherited from that of $B_d$ as the subset $x_k''=0$.
With these definitions, we obtain the main result of this section:
\begin{theorem} \label{thm:tccp}
  We have 
  \[
    2\pi\sigma_I =  \ind  F = (-1)^{d-1}  \tdeg (\rh^d) = (-1)^{\frac 12 d(d+1)+1} \deg (\rh^k).
  \]
\end{theorem}
In other words, $\ind F= \deg (\rh^k)$ when $d=1,2 \mod 4$ and $\ind F=-\deg (\rh^k)$ when $d=3,4 \mod 4$.

%As a summary of the proof, we write (i) the FH formula in terms of $\sigma_d$; (ii) the FH formula in terms of $\rh^d$; (iiii) using standard results on spheres, we identify the FH formula with the degree of $\rh^d$ on the sphere $\Sm_R$; (iv) this is identified with the degree ${\rm deg}(\rh^d,\bar B_d,0)$; (v) we then see $\rh^d=(\rh^k, \tilde \rh)$ with $\tilde \rh$ the augmentation map and that the degree of $\rh^d$ is the product of the other degrees. Now the degree of $\tilde\rh$ is one and this gives the result. 

The rest of this section is devoted to the proof of the theorem. Its main steps are as follows: (i) write \eqref{eq:FH} in terms of $\sigma_d$; (ii) next in terms of $\rh^d$; (iii) identify \eqref{eq:FH} with the degree of $\rh^d$ on the sphere $\Sm_R$; (iv) identify it with the degree ${\rm deg}(\rh^d,\bar B_d,0)$; (v) decompose $\rh^d=(\rh^k, \tilde \rh)$ with $\tilde \rh$ the augmentation map. The degree of $\rh^d$ is then the product of the other two degrees. Now the degree of $\tilde\rh$ is one and this gives the result. 

\begin{lemma} \label{lem:indad}
We have:
\[
  \ind F = \frac{-1}2    \dfrac{-(d-1)!}{(2\pi i)^d (2d-1)!}  \dint_{\Sm_R^{2d-1}} {\rm tr} \gamma_0^d(a_d^{-1}da_d)^{2d-1}. 
\]
\end{lemma}
\begin{proof}
  We have $a=a_{d-1}-im_d$ and $a_d=\sigma_+\otimes a+ \sigma_-\otimes a^*$. Therefore $a_d^{-1}da_d={\rm Diag}(a^{-*} da^*,a^{-1}da)$ and hence $(a_d^{-1}da_d)^{2d-1}={\rm Diag}((a^{-*} da^*)^{2d-1},(a^{-1}da)^{2d-1})$.  Thus with $\gamma^0_d=\sigma_3\otimes I$,
\[
    {\rm tr} \gamma^0_d (a_d^{-1} da_d)^{2d-1}  = {\rm tr} (a^{-*} da^*)^{2d-1} - {\rm tr} ( a^{-1}da)^{2d-1}.
\]
The above traces are not necessarily linearly dependent.
However, by Theorem \ref{thm:FH},  the integral  of the second term gives $-\ind F$ while the integral of the first term gives $\ind F^*=-\ind F$ since $a^*$ is the symbol of $F^*$. This gives the result.
\end{proof}

\begin{lemma}
  Let $a_d$ as above. Then for $|X|$ large enough,
\[
   {\rm tr} \gamma^0_d (a_d^{-1} da_d)^{2d-1} = (-1)^{d-1} |a_d|^{-2d} {\rm tr}  \gamma^0_d a_d (da_d)^{2d-1}.
\]  
\end{lemma}
\begin{proof}
  We first observe that $a_j^2=|\rh^j|^2$ for $k\leq j\leq d$  are proportional to identity thanks to \eqref{eq:clifcomm}. Since $a_d^2=|h^d|^2$ is scalar, then $a_d^{-1}=wa_d$ for $w=a_d^{-2}$ scalar so that 
\[
  a_d^{-1} da_d = -da_d^{-1} a_d = -dw a_d^2 - wda_d a_d
\]
and hence
$
  (a_d^{-1}da_d)^2 =  -dw a_d da_d - w (da_d)^2.
$
Using $dw da_d^2=0$ so that $dw da_d a_d=-dwa_d da_d$, we find
$
  (a_d^{-1}da_d)^4 =  dw^2 a_d (da_d)^3 +w^2 (da_d)^4
$
and more generally
\[
  (a_d^{-1}da_d)^{2n} = (-1)^n  [dw^n a_d (da_d)^{2n-1} +w^2 (da_d)^{2n}]
\]
as well as (for $d\geq2$)
\[
  (a_d^{-1}da_d)^{2d-1} = (-1)^{d-1}  [dw^{d-1}  (da_d)^{2d-2} +w^da_d (da_d)^{2d-1}].
\]
Thus we obtain for $d\geq1$ that
\[
   {\rm tr} \gamma^0_d  (a_d^{-1}da_d)^{2d-1}  = (-1)^{d-1}  |a_d|^{-2d} {\rm tr} \gamma^0_d a_d (da_d)^{2d-1}
\]
with the term in $dw$ vanishing since it involves the trace of a product of an even number of necessarily different (because of the product of exterior differentiations) gamma matrices.  Such traces necessarily vanish for Clifford matrices.
\end{proof}
\begin{lemma}\label{lem:adhd}
  For $a_d=\rh^d\cdot\Gamma_d$, we have
\[
   {\rm tr}  \gamma_d^0 a_d (da_d)^{2d-1} = (2i)^d \dsum_{\rho\in\mS_{2d}}  (-1)^\rho \rh^d_{\rho_1} d  \rh^d_{\rho_2} \wedge \ldots \wedge d\rh^d_{\rho_{2d}}, \quad (2i)^d \epsilon_{i_1,\ldots i_{2d}} = {\rm tr} \gamma^0_d \gamma^{i_1}_d \ldots \gamma^{i_{2d}}_d. 
\]
\end{lemma}
We recall that $\mS_{n}$ the set of permutations of $\{1,\ldots,n\}$.
The proof of the lemma directly comes from the construction of Clifford matrices  (and their orientation) and generalizes that ${\rm tr} \sigma_3 \sigma_1\sigma_2=2i$. The above three lemmas show that the index of $F$ is related to an appropriate integral of $\rh^d$.
\begin{lemma}\label{lem:inddeg}
 We have:
 \[
   \ind F = (-1)^{d-1} \tdeg \rh^d.
 \]
\end{lemma} 
\begin{proof}
Let $\Sigma$ be a smooth hypersurface in $\Rm^{2d}$ locally parametrized by $X=X(u)$ for $u\in \Rm^{2d-1}$. We introduce the $2d\times 2d$ matrix $L(u)$ constructed as follows (see \cite[Corollary 14.2.1]{DFN-SP-1985}). The first row is $L_{1i}(u)=\rh^d_i \circ X(u)$ while the following rows are $L_{j+1,i}(u) =\partial_{u_j} \rh^d_i \circ X(u)$ for $1\leq j\leq 2d-1$. We then observe that 
\begin{equation}\label{eq:Lu}
  \dsum_{\rho\in\mS_{2d}} (-1)^\rho \rh^d_{\rho_1} d  \rh^d_{\rho_2} \wedge \ldots \wedge d\rh^d_{\rho_{2d}} = (2d-1)!\  {\rm Det }L(u)  du_1\wedge \ldots \wedge du_{2d-1}.
\end{equation}
Indeed, with $\rh'$ the vector $(\rh_{\rho_2},\ldots, \rh_{\rho_{2d}})$ and $\nabla_u\rh'$ the Jacobian matrix, we have
\[
  d  \rh^d_{\rho_2} \wedge \ldots \wedge d\rh^d_{\rho_{2d}}  = {\rm Det} \nabla_u \rh' du_1\wedge \ldots \wedge du_{2d-1}.
\]
For any permutation $\rho$, we observe that $(-1)^\rho$ and ${\rm Det} \nabla_u \rh'$ change signs together so that denoting by $L_{\rho_1}$ the matrix $L$ with the first row and the $\rho_1$ column deleted, we have 
\[
  \dsum_{\rho\in\mS_{2d}} (-1)^\rho \rh^d_{\rho_1} d  \rh^d_{\rho_2} \wedge \ldots \wedge d\rh^d_{\rho_{2d}} = (2d-1)! \dsum_{\rho_1=1}^{2d} (-1)^{\rho_1}  {\rm Det} L_{\rho_1} du_1\wedge \ldots \wedge du_{2d-1},
\]
which gives \eqref{eq:Lu}, noting that $(2d-1)!$ is the number of permutations in $\mS_{d-1}$.

Collecting the above results we obtained that 
\[
  {\rm tr} (a^{-1} da)^{2d-1} = \frac12(-2i)^d (2d-1)!  |\rh^d|^{-2d} {\rm Det} L du_1 \wedge\ldots \wedge du_{2d-1}.
\]
Therefore,
\[
   \dfrac{-(d-1)!}{(2\pi i)^d (2d-1)!}  {\rm tr} (a^{-1}da)^{2d-1}  = \dfrac{(-1)^{d-1}}{\gamma_{2d-1}}  |\rh^d|^{-2d} {\rm Det} L du_1 \wedge\ldots \wedge du_{2d-1},\ \gamma_{2d-1} = \frac{2\pi^d}{(d-1)!}
\]
where $\gamma_{2d-1}$ is the volume of the unit sphere $\Sm^{2d-1}$.

Let $f$ be the Gauss map associated to $\rh^d$ and given by $f(X)=|\rh^d(X)|^{-1} \rh^d(X)$ for $X\in \Sm_R^{2d-1}$.  Then we recognize in the integration of the above term over $\Sm_R^{2d-1}$ the degree of $f$ \cite[Corollary 14.21]{DFN-SP-1985}. Moreover, the degree of the Gauss map $f$ is given in \cite[Theorem 14.4.4]{DFN-SP-1985} precisely by the sum in \eqref{eq:indexsum} and so equals $\deg (\rh^d;\bar B_d,0)$ when $0$ is a regular value of $\rh^d$. When $0$ is not a regular value, we appy the result of $\rh^d-y_0$ for $y_0$ small with result independent of $y_0$ as recalled above; see also \cite[Remark 1.5.10]{nirenberg1974topics}. With the chosen orientation to define $\tdeg$ and Theorem \ref{thm:FH}, we thus obtain the result of the lemma.
\end{proof}

\begin{lemma}
  We have  $\deg (\rh^d) = \deg (\rh^k)$.
\end{lemma}
\begin{proof}
%   The proof uses results on degree theory that can be found in \cite{nirenberg1974topics}.
%
%   We first observe that ${\rm deg}(\rh^d, \Sm_R^{2d-1},\Sm_R^{2d-1})={\rm deg} (\rh^d,\bar B_d,0)$ since $0$ is not in the range of $\rh^d$ restricted to $\Sm_R^{2d-1}$ the boundary of $B_d$. 
%   
%Here, we have defined for $\zeta\mapsto \rh(\zeta)$ a sufficiently smooth map from $\Rm^n$ to $\Rm^n$ such that $|\rh(\zeta)|\geq c_0>0$ for $\zeta\in C$ an open subset in $\Rm^n$ the degree
%\[
%  {\rm deg} (\rh, \bar C,0)= \dsum_{\zeta\in \rh^{-1}(y_0)} {\rm sgn } \ {\rm det}\ J_{\rh}(\zeta)
%\]
%where $y_0$ is a regular value of $\rh$, which exists by the Sard theorem, where the above set $\rh^{-1}(y_0)$ is finite, and where $J_{\rh}$ is the Jacobian matrix of $\rh$; see  \cite{DFN-SP-1985} and  \cite{nirenberg1974topics}.
%
Let $\zeta_j\mapsto \rh_j(\zeta_j)$ for $j=1,2$ be two smooth functions from $\Rm^{n_j}$ to itself with $|\rh_j(\zeta_j)|\geq c_0>0$ for $|\zeta_j|\geq R_j$. Let $B_j=B_j(0,R_j)$ be the centered balls of radius $R_j$ for the Euclidean metric in $\Rm^{n_j}$ for $j=1,2$. Let now $(\zeta_1,\zeta_2)=\zeta\mapsto \rh(\zeta)$ be the function from $\Rm^{n}$ to itself with $n=n_1+n_2$ defined by
\[
  \rh(\zeta) = (\rh_1(\zeta_1),\rh_2(\zeta_2)).
\]
We find that $|\rh(\zeta)|\geq c_0>0$ for $\zeta\in \partial C$ where $C=B_1\times B_2$. Since $0$ does not belong to the range of $\rh$ or $\rh_j$ on the respective boundaries, we can define the degrees
\[
  {\rm deg}(\rh,\bar C,0)\qquad \mbox{ and } \qquad  {\rm deg}(\rh_j,\bar B_j,0),\ \ j=1,2.
\] 
Let $y_0$ be a regular value of $\rh$, i.e., a point in $\rh(\bar C)\backslash \rh(\partial C)$ such that $\rh^{-1}(y_0)=\{\zeta\in \bar B_R;\ \rh(\zeta)=y_0\}$ are a finite number of isolated regular points (where $\nabla\rh$ is invertible). Note that $y_0=(y_1,y_2)$ with $y_j\in B_j$. By Sard's theorem, regular values exist. Then independently of such a $y_0$, 
\[
  {\rm deg}(\rh,\bar C,0) =\dsum_{\zeta\in \rh^{-1}(y_0)} {\rm sgn\  det\ } J_\rh(\zeta)
\]
where $J_\rh$ is the non-vanishing Jacobian of the map $\zeta\to \rh(\zeta)$. Now by construction,
\[
  {\rm det\ } J_\rh(\zeta) = {\rm det\ } J_{\rh_1}(\zeta_1) \ {\rm det\ } J_{\rh_2}(\zeta_2) .
\]
Moreover, $\rh(\zeta)=y_0$ means $(\rh_1(\zeta_1),\rh_2(\zeta_2))=(y_1,y_2)$ so that $\rh^{-1}(y_0)= \rh_1^{-1}(y_1) \times  \rh_2^{-1}(y_2)$ and hence
\[
  \dsum_{\zeta\in \rh^{-1}(y_0)} {\rm sgn\  det\ } J_\rh(\zeta) = \Big(\dsum_{\zeta_1\in \rh_1^{-1}(y_1)} {\rm sgn\  det\ } J_{\rh_1}(\zeta_1)\Big)\Big( \dsum_{\zeta_2\in \rh_2^{-1}(y_2)} {\rm sgn\  det\ } J_{\rh_2}(\zeta_2)\Big).
\]
We recognize the product of degrees for the regular values $y_j\in \rh_j(\bar B_j)\backslash \rh_j(\partial B_j)$. Since degrees are independent of such regular values locally, we obtain that 
\[
   \deg(\rh,\bar C,0) = \prod_{j=1}^2  \deg(\rh_j,\bar B_j,0).
\] 
We observe that $C\subset B_R$ the ball of radius $R=\sqrt{R_1^2+R_2^2}$. Since $|\rh|\geq c_0>0$ on $\bar B_ R \backslash C$, invariance of the results with respect to (continuous) domain changes (see \cite[Proposition 1.4.4]{nirenberg1974topics}) show that 
\begin{equation}\label{eq:eqdegd}
  {\rm deg}(\rh,\bar C,0) = {\rm deg}(\rh,\bar B_R,0)  = \prod_{j=1}^2  {\rm deg}(\rh_j,\bar B_j,0).
\end{equation}

We know apply the above construction to $\rh_1=\rh^k$ and $\rh_2$ the vector so that $\rh^d=(\rh^k,\rh_2)$ with $n_1=d+k$ and $n_2=d-k$ where $\Rm^{2d}$ is oriented using $d\xi_1\ldots d\xi_d dx_1 \ldots dx_d>0$ and the subspaces $\Rm^{d+k}$ (for $\rh^k$) and $\Rm^{d-k}$ (for $\rh_2$) with the induced orientation.  We observe that the degree of $\rh_2=\rh_2(x_k'')=(m_{k+1},\ldots, m_d)(x_k'')$ equals $1$ since the only point in $\rh_2^{-1}(0)=0$ and the Jacobian is identity there with the above orientation. Using \eqref{eq:eqdegd} and the definitions \eqref{eq:degrees} and \eqref{eq:degk} proves the result. 
\end{proof}
The above lemmas together with the change of orientation relation \eqref{eq:changeorient} conclude the proof of Theorem \ref{thm:tccp}.

\section{Applications}
\label{sec:appli}

The classification presented in section \ref{sec:local} applies to Hamiltonians in class A (Hermitian symbols) and AIII (Hermitian symbols with chiral symmetry) in arbitrary dimension. The Fedosov-\horm formula \eqref{eq:FH} shows that the index is controlled by the symbol $a$ of $F$, and hence that of $H_k$, restricted to any sphere with a sufficiently large radius $R$. This implies that the topological charge is independent of the symbol $a$ in the complement of that ball.  The main assumptions to apply \eqref{eq:FH} are that: (i) the symbol $a$ of $F$ is uniformly invertible for $|X|\geq R$ for some $R>0$, in which case: (ii) the topological charge solely depends on $a$ restricted to the sphere $|X|=R$. 

The theory of section \ref{sec:local} applies only to operators whose symbols satisfy the ellipticity constraint \eqref{eq:ellip}, which combined with the growth condition \eqref{eq:Sk} imply that the symbol $a$ grows to infinity as $|X|\to\infty$. This should be contrasted to the two-dimensional results  in \cite{bal2022topological,QB-NUMTI-2021} where the domain wall $m(x_1)$ is assumed to be bounded and constant away from a compact domain. 

For any symbol $a_k$ such that (i) and (ii) hold, we allow for the following modifications of the symbol $a_k$ in order to apply the theory of section \ref{sec:local}. Let $\eps>0$ and $r\to \aver{r}_\eps$ a smooth non-decreasing function from $\Rm_+$ to $\Rm_+$ such that 
\begin{equation}\label{eq:avereps}
   \aver{r}_\eps = \left\{ \begin{array}{rl} 1 & \qquad \eps r\leq 1 \\ \eps r & \qquad \eps r \geq2. \end{array}\right.
\end{equation}
We use the same notation for the smooth function $\Rm^p \ni y\to \aver{y}_\eps:=\aver{|y|}_\eps$. This function has the same leading asymptotic behavior as $\aver{\eps y}$ for $|y|\to\infty$. We consider the above regularization for $y$ being one or several of the variables in $X$. Such modifications of $a_k$ preserve (i)-(ii) and allow us to satisfy  \eqref{eq:ellip} as well as \eqref{eq:Sk} so that the theory of section \ref{sec:local} applies. %Note that when $\eps\ll1$, the regularization changes the coefficients of the Hamiltonian for $X$ large 

Consider for instance the regularized `Dirac' operator $H_1=D_1\sigma_1+D_2\sigma_2 + (m-\eta D\cdot D)\sigma_3$ with here $D\cdot D=D_1^2+D_2^2$ the (positive) Laplacian and $\Rm\ni \eta\not=0$. The definition of a bulk invariant is ambiguous when $\eta=0$ while it yields a Chern number $\frac12 (\sgn m + \sgn \eta)$ when $\eta\not=0$ \cite{B-BulkInterface-2018,bernevig2013topological}. As indicated above, this paper does not consider bulk invariants but rather topological charges and interface invariants that may be related to bulk-difference (rather than bulk) invariants \cite{bal2022topological}. To define a topology in the class of symbols analyzed in this paper and satisfy \eqref{eq:ellip}, we modify the Hamiltonian as
\[
  H_1 = D_1\sigma_1+D_2\sigma_2 + (\mu(x_1)-\eta \aver{D}_\eps^{-1} D\cdot D)\sigma_3
\]
where we assume that $\mu(x_1)$ equals $x_1$, say, outside of a compact set in $\Rm$. We then verify that the topological charge of $H_1$ equals $1$ and is independent of $\eps$ and of $\eta$ as expected since $\eta$ affects the bulk invariants but not the bulk-difference invariant \cite{B-BulkInterface-2018}.  Note that (i) and (ii) now hold with $m=1$. Alternatively, we can introduce $H_1=\aver{D}_\eps(D_1\sigma_1+D_2\sigma_2) + (\aver{x_1}_\eps x_1 - \eta D\cdot D)\sigma_3$ satisfying (i) and (ii) with now $m=2$.

We next consider several examples of topological insulators and superconductors in dimensions $d=1,2,3$ \cite{bernevig2013topological,RevModPhys.83.1057,sato2017topological,schindler2018higher,volovik2009universe} where the theories of both sections \ref{sec:local} and \ref{sec:deg} apply. While the theory leading to Theorems \ref{thm:tcc} and \ref{thm:tccp} applies to a large class of practical settings, there are important exceptions, in particular the $3\times3$ Hamiltonian \eqref{eq:water} describing fluid waves (before it is appropriately regularized) and the scalar Hamiltonian \eqref{eq:magsch} that appears in the analysis of the integer quantum Hall effect; these cases will be treated in more detail below. We also refer to \cite{bal2022multiscale} for an application to Floquet topological insulators where a variation on Theorem \ref{thm:tcc}  is used to compute invariants for operators that are not in the form \eqref{eq:cliffordrep}.

\paragraph{Dirac operator.}
The first example is the Dirac operator with $H_0=\ow a_0$ for $a_0(X)=\rh^0(X)\cdot\Gamma_0\in S^1_0(g^s)$ in dimension $d$, where
\[
  \rh^0(X)= (\xi_1,\ldots,\xi_d)
\]
and $\Gamma_0$ are Clifford matrices acting on spinors in $\Cm^{2^{\kappa_0}}$ with $\kappa_0=\lfloor \frac d2 \rfloor$. These generalize the cases $d=1,2,3$ considered in the introduction.  We then observe from \eqref{eq:indexsum} that  $\deg \rh^0=1$ since $(\rh^0)^{-1}(0)=\{0\}$ and $\nabla \rh^0(0)=I_d$, and that the topological charge of $H_0$ is given by $\ind F=2\pi\sigma_I(H_{d-1})=(-1)^{\frac12 d(d+1)+1}$, i.e., $\ind F=1$ in dimensions $1,2\mod 4$ and $\ind F=-1$ in dimensions $3,4\mod 4$.

The above topological charge is multiplied by $\sign \det A$ if $A$ is a non-singular (constant) matrix in $\Mm_d(\Rm)$ and $\rh^0(\xi)$ above is replaced by  $A\rh^0(\xi)$. 

The topological charge is given equivalently by the degree of $\rh^0$ or by that of $\rh^d=(\xi_1,\ldots,\xi_d,x_1,\ldots x_d)$ with $\deg \rh^0=\deg \rh^d=1$.

\medskip

The topological charge is also stable against large classes of (smooth) perturbations of arbitrary amplitude. For instance if $\rh^0_j$ for $1\leq j\leq d$ is replaced by $\tilde \rh^0_j(X)=b_j(x) \rh^0_j(\xi)$ for $b_j(x)$ smooth and equal to $1$ outside of a compact set in $\Rm^d$, then the corresponding symbol $\tilde\rh^0\cdot\Gamma_0\in S^1_0(g^s)$ though not necessarily in $S^1_0(g^i)$. Note that a more isotropic perturbation of the form  $b_j(X) \rh^0_j$ for $b_j(X)$ smooth and equal to $1$ outside of a compact set in $\Rm^{2d}$ would generate a perturbation in $S^1_0(g^i)$ although one that is no longer a differential operator. This illustrates the reason why we considered the (reasonably large) classes $S^m_k(g^s)$. 

The model Hamiltonian in the presence of one domain wall is $\rh^1(X)=(\xi_1,\ldots,\xi_d,x_1)$. Its topological charge is then again computed based on $\deg \rh^1=1$. Domain walls of the form $b_j(x)x_1$ even with $b_j(x)=1$ outside of a compact domain no longer necessarily generate perturbations such that $a_1$ remains in $S^1_1$ and are therefore not allowed in the theoretical framework of this paper. We may however replace $x_1$ by $m(x_1)$ equal to $x_1$ outside of a compact set. Inside that compact set, the level set $m(x_1)=0$ is then arbitrary. 

For a time-dependent picture on how wavepackets propagate along curved interfaces for two-dimensional Dirac equations, see also \cite{bal2022magnetic,bal2021edge}.

\medskip

For concreteness and illustration, we spell out the details of the calculations and comparisons used in Theorems \ref{thm:FH} and \ref{thm:tccp} when $d=1$. We then have $a_0=\xi=\rh^0$ while $a=\xi-ix$ and $a_1=\xi\sigma_1+x\sigma_2$. We then observe that $a^{-1}da=(\xi^2+x^2)^{-1}(\xi d\xi+x dx + i(\xi dx-xd\xi))$. In polar coordinates $\xi=r\cos\theta$ and $x=r\sin\theta$, we observe that $a^{-1}da=r^{-1}dr-id\theta$ whose integral along the curve $r=1$ gives $-2\pi i$ and hence $\ind \ow a=1$ as a direct application of the Fedosov-\horm formula \eqref{eq:FH}. 

We now observe that the index may be computed as in Lemma \ref{lem:indad} from $a_1=\xi\sigma_1+x\sigma_2$ with ${\rm tr} \sigma_3 a_1^{-1}da_1= a^{-*}da^*-a^{-1}da$ and $a^*=\xi+ix$ so that $a^{-*}da^*=r^{-1}dr+id\theta$. This shows that ${\rm tr} \sigma_3 a_1^{-1}da_1= 2id\theta$ whose appropriate integral gives the topological charge. Now, $a_1^{-1}=|\rh^1|^{-2}(\rh^1_1\sigma_1+\rh^1_2\sigma_2)$ for $\rh^1=(\xi,x)$.  Thus, as in Lemma \ref{lem:adhd}, ${\rm tr} \sigma_3 a_1^{-1}da_1=|\rh^1|^{-2}{\rm tr} \sigma_3 (\rh^1_1\sigma_1+\rh^2_2\sigma_2)(d\rh^1_1\sigma_1+d\rh^1_2\sigma_2)=2i(\rh^1_1d\rh^1_2-\rh^1_2d\rh^1_1)$.  We recognize in the integral of the latter form over the circle an expression for the degree of $\rh^1$ written as the degree of the Gauss map which to $X\in \Sm^1$ associates $\rh^1(X)/|\rh^1(X)|$. Using the expression \eqref{eq:indexsum} of the degree over the unit disc $C$ gives $\deg \rh^1=1$ since $\nabla \rh^1=I_2$ at the unique point $X=0$ where $\rh^1=0$. 

\medskip

The above orientation of the vector fields $\rh^k$ is natural in the context of topological insulators or superconductors, which are typically first written for spatially-independent coefficients. A different orientation helps to better display the invariance of the indices of Dirac operators across spatial dimensions (see \cite[Proposition 19.2.9]{H-III-SP-94} for a related construction).  We start with $F_1=D_1-ix_1=-i \fa_1$ and then define iteratively
\[
  F_{n+1} = \sigma_-\otimes F_n + \sigma_+ \otimes F_n^* + \sigma_3\otimes D_{n+1} -i x_{n+1}.
\]
The above construction is an example of the more general structure
\[
  f\sharp g := \begin{pmatrix} f\otimes I & I \otimes g^* \\ I \otimes g & -f^* \otimes I \end{pmatrix}
\]
where we verify that  $\ind{ f\sharp g} =\ind f \, \ind g$. We apply it with $g=F_{n-1}$ and $f=D_n-ix_n$. It is then straightforward to obtain that $\ind F_n=1$ for all $n\geq1$. We then observe that $F_n (1,0,\ldots 0)^t e^{-\frac12 |x|^2}=0$ with spinor $(1,0,\ldots 0)^t$ of dimension $2^{n-1}$.

%%%
\paragraph{Dirac operator with magnetic field.} We now incorporate constant magnetic fields at infinity for magnetic potentials written in an appropriate gauge. Let us consider the case $d=2$ for concreteness and the (minimal coupling) operator
\[
  H_0 = (D_1-A_1)\sigma_1 + (D_2-A_2)\sigma_2 +V
\]
with $A=(A_1,A_2)$ the magnetic vector potential and $V$ a bounded scalar potential with compact support, say. The magnetic field is given by $B=\nabla\times A=\partial_1 A_2-\partial_2 A_1$. We choose the gauge such that $A_2=B_0x_1+\tilde A_2$ and $A_1=\tilde A_1$ for $\tilde A$ an arbitrary (smooth) compactly supported perturbation. In that gauge, we obtain that 
\[
  H_1 =  (D_1-A_1)\sigma_1 + (D_2-A_2)\sigma_2 +V + x_1\sigma_3
\]
is an operator $H_1=\ow a_1$ with $a_1\in ES^1_1$ for $n_1=2$. Note that for $H_0=\ow a_0$, we do not have that $a_0$ belongs to $ES^1_0$ because of the presence of the unbounded magnetic potential. We would also not have that $a_1$ belongs to $ES^1_1$ if $A=(-\frac12 B_0 x_2,\frac12 B_0x_1)$ were chosen in the Landau gauge, for instance. While physical phenomena have to be independent of the choice of a gauge, the appropriate functional setting to handle constant magnetic fields, and hence unbounded magnetic potentials, is not. With the above construction, we obtain that $2\pi\sigma_I(H_1)=\ind F=1$ for $F=H_1-ix_2$ since the topological charge is given by
\[
  \deg (\xi_1,\xi_2-B_0x_1,x_1,x_2) = \deg (\xi_1,\xi_2,x_1,x_2)=1.
\]
We could more generally consider a magnetic field with constant values as $x_1\to\infty$, for instance with $A_2= B_0 \frac2\pi \arctan(x_1) x_1$. The topological charge of $H_1$ remains equal to $1$.
The magnetic field therefore has no influence on the topological charge in this setting.

\paragraph{Higher-order topological insulator.} Let us consider the Weyl operator $D\cdot\sigma$ in dimension $d=3$.  As we considered in the introduction, the operator $H_2=\sigma_1\otimes D\cdot\sigma+\sigma_2\otimes I x_1+\sigma_3\otimes I x_2$ generates a hinge in the third direction along which asymmetric transport is possible. With our choice of orientation, we have $2\pi\sigma_I(H_2) = - \deg (\xi_1,\xi_2,\xi_3,x_1,x_2)=-1$. 

By implementing more general domain walls, an arbitrary number of asymmetric modes may be obtained. This is done by considering for $p\in\Zm$,
\[
  H_2 := \sigma_1\otimes D\cdot\sigma+\sigma_2\otimes I_2 \,{\rm Re} (x_1+ix_2)^p +\sigma_3\otimes I_2 \,  {\rm Im} (x_1+ix_2)^p.
\]
We thus deduce from Theorem \ref{thm:tccp} that
\begin{equation}\label{eq:hotip}
  2\pi\sigma_I(H_2) = - \deg ({\rm Re} (x_1+ix_2)^p,{\rm Im} (x_1+ix_2)^p) = -p.
\end{equation}
The last result is most easily obtained by identifying, as we did in the proof of Lemma \ref{lem:inddeg}, the degree of $\rh^2$ on the unit ball with the degree of the Gauss map $x\to \hat \rh^2=\rh^2/|\rh^2|$ from the unit circle $\Sm^1$ to itself and then to the degree (winding number) of the map $x_1+ix_2\to (x_1+ix_2)^p$ from the unit circle to itself, which equals $p$.

By an appropriate construction of the coefficients in the Hamiltonian $H_2$ acting on $\Cm^4$, we thus obtain a low-energy model for a coaxial cable with an arbitrary number of asymmetric protected modes along the hinge \cite{schindler2018higher}.

\paragraph{Topological superconductors.} Several superconductors and superfluids \cite{bernevig2013topological,volovik2009universe} are modeled by Hamiltonians of the form 
\[
 H_1 = \big(\eta D\cdot D-\mu\big)\sigma_3\otimes I_2 + H_\Delta
\]
with coupling term $H_\Delta=\sum_{i,j=1,2} \Delta_{ij}(X) \sigma_i\otimes \sigma_j$ for scalar operators $\Delta_{ij}$ and $\eta=(2m^*)^{-1}$ for a mass of the quasi-particle $m^*>0$. For the above choice of the order parameter\footnote{We use $\Delta$ for the order parameter as is customary in the superconductor literature. The (positive) Laplace operator is denoted by $D\cdot D$.} $\Delta$, these Hamiltonians acting on $\Cm^4$ separate into two $2\times2$ Hamiltonians (acting on the first and fourth components, and the second and third components, respectively). We now consider several such examples in one and two space dimensions. 

\paragraph{One dimensional examples. }
For $d=1$, an example with the order parameter $\Delta$ proportional to $D_x$ gives
\[
  \tilde H_1 =  (\eta D_x^2-\mu)\sigma_3 + D_x{\rm Re}\Delta \sigma_1 - D_x {\rm Im}\Delta \sigma_2 
\]
with $0\not=\Delta\in\Cm$.  Let $\Delta=|\Delta|e^{i\theta}$ and $g=e^{i\frac\theta2\sigma_3}$. We then verify that
\[
  g\tilde H_1  g^* =  (\eta D_x^2-\mu)\sigma_3 + |\Delta| D_x \sigma_1
\]
and so we may assume $\Delta$ real-valued. Define $g_2=e^{i\frac\pi4 \sigma_2}$ and $g_1=e^{i\frac\pi4 \sigma_1}$.  We verify that  $g_1g_2 \sigma_{1,2,3}(g_1g_2)^*=\sigma_{3,1,2}$ so that 
\[
  g_1g_2g\tilde H_1 (g_1g_2g)^* =  (\eta D_x^2-\mu)\sigma_1 + \Delta D_x \sigma_2 = \tilde \rh^1 \cdot \Gamma_1,\quad \tilde \rh^1=(\eta \xi^2-\mu,\Delta \xi).
\]
This is of the form $\sigma_-\otimes F^* + \sigma_+\otimes F$ with $F= (\eta D_x^2-\mu) -i\Delta D_x$. 

 In order for $F$ to be a Fredholm operator, we need to introduce a domain wall. This may be achieved in two different ways: it may be implemented by either the chemical potential $\mu=\mu(x)$ or by the order parameter $\Delta=\Delta(x)$.
As we mentioned in the introduction, the symbol of $H_1$ has to be asymptotically homogeneous for the ellipticity condition \eqref{eq:ellip} to hold and the theories developed in the preceding sections to apply. We thus regularize the operator using functions of the form $\Rm^p\ni y\to \aver{y}_\eps=\aver{|y|}_\eps$ in \eqref{eq:avereps}. The regularization does not modify the symbol on compact domains in $\Rm^{2d}$ for $0<\eps$ sufficiently small and hence does not affect the computations of the index in \eqref{eq:FH} and \eqref{eq:degexplicit}. 

When $\eta>0$, we consider two regularized operators, one with a domain wall in the chemical potential
%\begin{equation}\label{eq:1dcpwall}
%   H_1 = (\eta D_x^2 - \mu x \aver{x}_\eps)\sigma_1 + \Delta D_x \aver{D_x}_\eps \sigma_2= \rh^1\cdot\Gamma_1,\quad \rh^1=(\eta \xi^2-\mu x \aver{x}_\eps, \Delta \xi \aver{\xi}_\eps),
%\end{equation}
%\tb{Two alternatives}
\begin{equation}\label{eq:1dcpwall}
   H_1 = (\eta \aver{D_x}_\eps^{-1} D_x^2 - \mu x )\sigma_1 + \Delta D_x  \sigma_2= \rh^1\cdot\Gamma_1,\quad \rh^1=(\eta \aver{\xi}_\eps^{-1} \xi^2-\mu x, \Delta \xi ),
\end{equation}
%with $X=(x,\xi)$
and one with a domain wall in the order parameter
\begin{equation}\label{eq:1dorderwall}
 H_1 = (\eta D_x^2 - \mu \aver{x}_\eps^2)\sigma_1 + \Delta \frac{D_x  x +xD_x}2 \sigma_2= \rh^1\cdot\Gamma_1,\quad \rh^1=(\eta \xi^2-\mu \aver{x}_\eps^2, \Delta \xi x ) .
\end{equation}
%\tb{We have another regularization in the introduction with $m=1$ instead of $m=2$. The former may be better.}
We observe that for $H_1=\ow a_1$, then $a_1\in ES^m_1$ is elliptic with $m=1$ in the first example and $m=2$ in the second example. Consider the second case \eqref{eq:1dorderwall}. We wish to show that $|\rh^1|^2\geq C(|X|-1)^4$. This is clear for $|x|\leq1$ and for $|x|\geq1$, then $|x|\geq C\aver{x}_\eps$ for $C>0$ so that $|\rh^1|^2\geq (\eta\xi^2-\mu\aver{x}^2_\eps)^2 + C\xi^2\aver{x}_\eps^2$ for $C>0$. The latter expression is homogeneous in $(\xi,\aver{x}_\eps)$ and non-vanishing on the unit sphere in these variables. This shows that $a_1\in ES_1^2$. A similar computation shows that $a_1=\rh^1\cdot\Gamma\in ES^1_1$ in \eqref{eq:1dcpwall}. Note that we could also have used the following regularization for the first example: $ \rh^1=(\eta \xi^2-\mu x \aver{x}_\eps, \Delta \xi \aver{\xi}_\eps)$, in which case $\rh^1\cdot\Gamma_1\in ES^2_1$.

We now compute the topological charges of the regularized operators starting with \eqref{eq:1dcpwall}. We observe that $\rh^1(\xi,x)$ vanishes only at $x=\xi=0$. The Jacobian there has determinant equal to $\mu\Delta$. The topological charge of  \eqref{eq:1dcpwall} is therefore equal to $\ind F=\deg \rh^1=\sgn{\mu\Delta}$ assuming $\mu\Delta\not=0$. Here and below, $F$ is defined as usual by the relation $H_1=\sigma_-\otimes F^* + \sigma_+\otimes F$.

We next turn to \eqref{eq:1dorderwall}, where $\rh^1(\xi,x)$ vanishes when $\eta\xi^2=\mu$ and $x=0$. When $\mu<0$, there is no real solution to this equation and the topological charge vanishes. When $\mu>0$, we have two solutions $\xi=\pm\sqrt{\mu/\eta}$. At these points,  the Jacobian matrix $\nabla \rh^1$ has components $(2\eta\xi;0;\Delta x;\Delta \xi)$ with determinant equal to $2\eta\Delta \xi^2$. The topological charge of $H_1$ in  \eqref{eq:1dorderwall} is therefore equal to  $\ind F=\deg \rh^1=2\sgn{\Delta}$.

Let us finally consider the asymptotic regime $\eta=0$ for a mass term $m^*\to\infty$ and a corresponding Hamiltonian $H_1=-\mu\sigma_1+\Delta D_x \sigma_2$. A domain wall in the chemical potential is then modeled by $\mu(x)=\mu x$. We then observe that $H_1=\ow a^1$ with $a^1\in ES^1_1$ and a topological charge equal to $\ind F=\deg \rh^1=\sgn{\mu\Delta}$ as in the setting $\eta>0$.  A domain wall in the order parameter requires the following regularized Hamiltonian $H_1 = - \mu \aver{x}_\eps^2\sigma_1 + \Delta \frac12(D_x  x+xD_x)  \sigma_2$, which is however gapped for $\mu\not=0$ and hence topologically trivial.

\medskip

The regularization of the above Hamiltonians is necessary for the symbol $a_1$ to have eigenvalues of order $|X|^m$ as $|X|\to\infty$.  The degree of the corresponding field $\rh^1$ is independent of the regularizing parameter $\eps$. The corresponding analysis with bounded domain walls, for differential operators and under suitable assumptions, is considered in \cite{bal2022topological,QB-NUMTI-2021}.

\paragraph{Two-dimensional examples.} We now consider two-dimensional examples of the above superconductor models. The $p+ip$ (or p-wave) model with order parameter proportional to momentum, is of the form 
\[
  H_1 = \Delta_1D_1 \sigma_1 + \Delta_2 D_2 \sigma_2 +  (\eta D\cdot D-\mu)\sigma_3.
\]
We assume here that $\Delta_1$ and $\Delta_2$ are real-valued. The case $\eta=0$ is a Dirac operator and was treated earlier. We thus assume $\eta>0$. A domain wall in the chemical potential is then implemented as 
%\begin{equation}\label{eq:2dcpwall}
%   H_1 =   \Delta_1 D_1 \aver{D_1}_\eps \sigma_1 +  \Delta_2 D_2\aver{D_2}_\eps \sigma_2 + (\eta D\cdot D- \mu x_1 \aver{x_1}_\eps)\sigma_3
%\end{equation}
\begin{equation}\label{eq:2dcpwall}
   H_1 =   \Delta_1 D_1  \sigma_1 +  \Delta_2 D_2 \sigma_2 + (\eta \aver{D}_\eps^{-1} D\cdot D- \mu x_1)\sigma_3.
\end{equation}
%\tb{Here again, it is probably simpler to use $m=1$ than $m=2$ using $\aver{D}_\eps^{-1}$.}
The symbol of this operator is $\rh^1\cdot\Gamma_1$  with  $\rh^1(\xi_1,\xi_2,x_1) =(\Delta_1 \xi_1 ,\Delta_2 \xi_2 , \eta \aver{\xi}_\eps^{-1} |\xi|^2-\mu x_1)$.
%This operator has for symbol $a_1=\rh^1\cdot\Gamma_1$ with 
%\[
%  \rh^1(\xi_1,\xi_2,x_1) = (\Delta_1 \xi_1 \aver{\xi_1}_\eps,\Delta_2 \xi_2 \aver{\xi_2}_\eps, \eta|\xi|^2-\mu x_1 \aver{x_1}_\eps).
%\]
This regularization ensures that $a_1\in ES^1_1$ is elliptic.  We could have defined a regularization in $ES^2_1$ instead with $\rh^1=(\Delta_1 \xi_1 \aver{\xi_1}_\eps,\Delta_2 \xi_2 \aver{\xi_2}_\eps, \eta|\xi|^2-\mu x_1 \aver{x_1}_\eps).$

It remains to compute the degree of $\rh^1$. We find that $\rh^1=0$ when $\xi_1=\xi_2=x_1=0$ and that the Jacobian determinant there is equal to $-\mu\Delta_1\Delta_2$. The topological charge of the above operator is therefore $2\pi\sigma_I(H_1)=\deg \rh^1 =-\sgn{\mu\Delta_1\Delta_2}$, which is consistent with \cite{bernevig2013topological}.

In \cite[Chapter 22]{volovik2009universe}, the domain wall is implemented in the order parameter $\Delta_1$, which after appropriate regularization, gives:
\[
  H_1 = \Delta_1 \frac{x_1D_1+D_1x_1}2 \sigma_1 + \Delta_2 D_2 \aver{D_2}_\eps \sigma_2 +  (\eta D\cdot D-\mu \aver{x_1}^2_\eps)\sigma_3.
\]
The constants $\Delta_1$, $\Delta_2,$ and $\mu$ are assumed not to vanish and $\mu>0$.
The symbol $a_1=\rh^1\cdot\Gamma_1\in ES^2_1$ is given by
\[
   \rh^1(\xi_1,\xi_2,x_1) = (\Delta_1 \xi_1 x_1,\Delta_2 \xi_2 \aver{\xi_2}_\eps, \eta|\xi|^2-\mu \aver{x_1}^2_\eps).
\]
We have $\rh^1=0$ when $\xi_2=0$, $x_1=0$, and $\eta\xi_1^2=\mu$. At each of the two solutions, the Jacobian of $\rh^1$ is given by
\[
  \nabla \rh^1 = \begin{pmatrix} 0 & 0 & \Delta_1 \xi_1 \\ 0 & \Delta_2  & 0 \\ 2\eta \xi_1 &0 & 0  \end{pmatrix},\qquad |\nabla \rh^1| = -2 \eta \Delta_1 \Delta_2 \xi_1^2 .
\]
Therefore, the topological charge of $H_1$ is equal to $2\pi \sigma_I(H_1) = \deg \rh^1= -2\sgn{\Delta_1\Delta_2}$ as in \cite[Chapter 22]{volovik2009universe}. When $\mu<0$, the find $\deg \rh^1=0$ again. 

\medskip

In \cite[Chapter 22]{volovik2009universe}, a model for a d-wave superconductor is given as 
\[
  H_1 =\Delta_1 D_1D_2 \sigma_1 + \Delta_2 (D_1^2-D_2^2) \sigma_2 + (\eta D\cdot D-\mu) \sigma_3.
\]
Following  \cite[Chapter 22]{volovik2009universe}, we implement a domain wall in $\Delta_1(x_1)$ and a regularization that gives the operator
\[
  H_1 = \ow \rh^1\cdot\Gamma_1,\quad \rh^1 = (\Delta_1 \frac{x_1}{\aver{x_1}_\eps} \xi_1\xi_2, \Delta_2(\xi_1^2-\xi_2^2), \eta|\xi|^2 - \mu \aver{x_1}_\eps^2).
\]
This generates a symbol $a_1=\rh^1\cdot\Gamma_1\in ES_1^2$ as may be verified. Then $\rh^1=0$ when $\xi_1^2=\xi_2^2=\frac{\mu}{2\eta}$ while $x_1=0$. At each of these four roots, we compute
\[
  \nabla \rh^1 = \begin{pmatrix} 0 & 0 & \Delta_1 \xi_1\xi_2 \\ 2\Delta_2 \xi_1 & -2\Delta_2 \xi_2 & 0 \\ 2\eta \xi_1 &2\eta \xi_2 & 0  \end{pmatrix},\qquad |\nabla \rh^1| = 4 \eta \Delta_1 \Delta_2 \xi_1^2 \xi_2^2.
\]
The sign of the Jacobian is the same at each of the roots so that $\deg \rh^1= 4 \sgn{\Delta_1\Delta_2}$. Therefore, we obtain a topological charge $2\pi\sigma_I(H_1)=\deg \rh^1 =4\sgn{\Delta_1\Delta_2}$ when both $\mu>0$ and $\eta>0$. When $\mu<0$, the operator is gapped and topologically trivial again.

%\tb{Unlike the picture in Volovik, we have different signs of the charge for the p- and d-wave superconductors. }

\paragraph{Three-dimensional example.} Following \cite[(17.24)]{bernevig2013topological}, we consider the time-reversal invariant superconductor (or superfluid) model
\[
  H_1 =  \Delta \sigma_1 \otimes (D\cdot\sigma) + (\eta D\cdot D-\mu) \sigma_3\otimes I
\]
acting on $\Cm^4$. When $\eta=0$ and $\mu(X)=\mu x_1$, we obtain a standard Dirac operator with a topological charge $\sgn{\mu \Delta}$ as may be verified (see also the following calculations). When $\eta>0$, we conjugate the above operator by $g_1\otimes I$ (which maps $\sigma_3$ to $-\sigma_2$) and after regularization and domain wall $\mu(X)=\mu x_1$ obtain the Hamiltonian
\[
  H_1 =  \Delta  \sigma_1 \otimes (D\cdot\sigma)  +  (\mu x_1-\eta \aver{D}_\eps^{-1} D\cdot D ) \sigma_2\otimes I.
\]
The operator has an elliptic symbol in $ES^1_1$. We can then introduce as earlier $H_2=H_1 + \sigma_3\otimes I x_2 $ and $F=H_2-ix_3$. Following Theorem \ref{thm:tccp}, the topological charge of $H_1$ is then defined as $\ind F = 2\pi\sigma_I(H_2) = -\deg \rh^1$ with 
\[
   \rh^1 = (\Delta  (\xi_1,\xi_2,\xi_3), \mu x_1 -\eta \aver{\xi}_\eps^{-1}|\xi|^2).
\]
We find $\rh^1=0$ at the point $\xi=0$ and $x_1=0$.  The Jacobian $\nabla \rh^1$ at this point has determinant $\Delta^3\mu$ so that the topological charge is given by $\ind F=-\sgn{ \mu\Delta}$.
\bigskip

\paragraph{Other Hamiltonians.} The above examples all fit within the framework of operators with symbols $a_k=\rh^k\cdot\Gamma_k$ verifying that $a_k^2$ is a scalar operator resulting in two energy bands. The ellipticity requirement is that the energies tend to infinity as $|X|$ goes to infinity with a prescribed power $m>0$. In this setting, the topological charge can conveniently be computed as the degree of the field $\rh^k$ as shown in the preceding examples. 

The computations easily extend to operators of the form $H_k\oplus \tilde H_k$ or more general direct sums of operators that are in the above form. More generally, the topological charge conservation result in Theorem \ref{thm:tcc} applies to operators  beyond those of the form $a_k=\rh^k\cdot\Gamma_k$ provided that the symbol has eigenvalues appropriately converging to $\infty$ as $|X|\to\infty$. 

The topological charge conservation in Theorem \ref{thm:tcc} also applies to the sequence of effective Hamiltonians one obtains for continuous models of two-dimensional Floquet topological insulators. Such effective Hamiltonians are not in the form \eqref{eq:cliffordrep} and their asymmetric transport properties are most easily estimated by bulk-difference invariants related to the Fedosov-\horm formula; see \cite{bal2022multiscale}. 

There are natural examples of topologically non-trival Hamiltonians to which the theory presented in this paper does not apply directly. A typical example is based on the shallow water wave (two-dimensional) Hamiltonian \cite{delplace2017topological,souslov2019topological}
\begin{equation}\label{eq:water}
  H_0 = \begin{pmatrix} 0 &D_1 & D_2 \\ D_1 & 0 & -if \\ D_2 & if & 0 \end{pmatrix} = \ow a_0,\qquad a_0=  \begin{pmatrix} 0 &\xi_1 & \xi_2 \\ \xi_1 & 0 & -if \\ \xi_2 & if & 0 \end{pmatrix}
\end{equation}
where $f=f(x_1)$ represents a (real-valued) Coriolis force. The symbol of that operator has two eigenvalues $\pm\lambda(x,\xi)$ with $\lambda(x,\xi)=\sqrt{\xi_1^2+\xi_2^2+f^2(x_1)}$ similar to those of a Dirac operator and a third uniformly vanishing eigenvalue. Therefore $H_0+\alpha$ is gapped for $\alpha\not=0$ but with a gap independent of $\xi$. The presence of this flat band of essential spectrum creates difficulties that are not only technical: the topological charge conservation (a bulk-interface correspondence in dimension $d=2$) in Theorem \ref{thm:tcc} does not always hold although it does for certain profiles $f(x)$; see \cite{bal2022topological,graf2020topology,tauber2019bulk}. 

A regularized version of the above Hamiltonian, however, fits into the framework of the current paper. We observe that the kernel of $a_0$ is associated to the eigenvector $\psi_0=(\xi_1^2+\xi_2^2+f^2)^{-\frac12} (if,\xi_2,-\xi_1)^t$. We define the projector $\Pi_0=\psi_0\psi_0^*$ and for $0\not=\mu\in\Rm$ the regularized (pseudo-differential) Hamiltonian $H_\mu=H_0+ \mu \ow \lambda^2(1+\lambda^2)^{-\frac12} \Pi_0$. The symbol now has eigenvalues given by $\pm\lambda$ and $\mu \lambda^2(1+\lambda^2)^{-\frac12}$ (ensuring that the symbol of $H_\mu$ is smooth). If we choose $f(x_1)=\nu x_1$ with $\nu\not=0$ to generate a domain wall in the first variable, we observe that $H_\mu$ is elliptic with symbol in $ES_1^1$ (with $m=1$). Following computations in, e.g., \cite{bal2022topological,QB-NUMTI-2021}, which we do not reproduce here, we find that the topological charge of $H_\mu$ equals $2\sgn{\nu}$ independently of the choice of $\mu\not=0$. This is the topological charge obtained when $\mu=0$ under smallness constraints in \cite{bal2022topological}.

A second type of Hamiltonian for which the theory cannot possibly apply is the ubiquitous two-dimensional scalar magnetic \schr operator
\begin{equation}\label{eq:magsch}
  H = (D_1-A_1)^2 + (D_2-A_2)^2 +V,
\end{equation}
for instance for $A=(0,Bx_1)$ so that $\nabla\times A=B$ is a constant magnetic field. In this case, the spectral decomposition of this operator gives rise to a countable number of infinitely degenerate flat bands, the Landau levels, which are incompatible with the elliptic structure we impose on the symbol of the Hamiltonian in this paper. Even the notion of a domain wall is not immediate for the above model. Note that the integral in \eqref{eq:FH} vanishes for $a$ scalar-valued in dimension $d\geq2$ since then $(a^{-1}da)^{\wedge 2}=0$. For the numerous applications of this model, both discrete and continuous, to the understanding of the integer quantum Hall effect, we refer the reader to, e.g., \cite{PhysRevLett.65.2185,avron1994,bellissard1994noncommutative,combes2005edge,dombrowski2011quantization,elbau2002equality,hatsugai1993chern,schulz2000simultaneous}.

%Examples for which the theory does not apply are the $3\times3$ system of topological water waves and the scalar \schr\ equation with magnetic field modeling the integer quantum Hall effect.

\section*{Acknowledgment} This work was funded in part by NSF Grants DMS-1908736 and EFMA-1641100.

%
%
%%%%%%%%%%%%%%%%%
{\small
%\bibliography{../../../bibliography,../bibTI} 

\begin{thebibliography}{10}

\bibitem{PhysRevLett.65.2185}
{\sc J.~E. Avron, R.~Seiler, and B.~Simon}, {\em {Quantum Hall effect and the
  relative index for projections}}, Phys. Rev. Lett., 65 (1990),
  pp.~2185--2188.

\bibitem{avron1994}
\leavevmode\vrule height 2pt depth -1.6pt width 23pt, {\em Charge deficiency,
  charge transport and comparison of dimensions}, Comm. Math. Phys., 159
  (1994), pp.~399--422.

\bibitem{AVRON1994220}
\leavevmode\vrule height 2pt depth -1.6pt width 23pt, {\em The index of a pair
  of projections}, Journal of Functional Analysis, 120 (1994), pp.~220 -- 237.

\bibitem{B-BulkInterface-2018}
{\sc G.~Bal}, {\em Continuous bulk and interface description of topological
  insulators}, Journal of Mathematical Physics, 60 (2019), p.~081506.

\bibitem{B-EdgeStates-2018}
\leavevmode\vrule height 2pt depth -1.6pt width 23pt, {\em Topological
  protection of perturbed edge states}, Communications in Mathematical
  Sciences, 17 (2019), pp.~193--225.

\bibitem{bal2022topological}
\leavevmode\vrule height 2pt depth -1.6pt width 23pt, {\em Topological
  invariants for interface modes}, Communications in Partial Differential
  Equations,  (2022), pp.~1--44.

\bibitem{bal2022magnetic}
{\sc G.~Bal, S.~Becker, and A.~Drouot}, {\em Magnetic slowdown of topological
  edge states}, arXiv preprint arXiv:2201.07133,  (2022).

\bibitem{bal2021edge}
{\sc G.~Bal, S.~Becker, A.~Drouot, C.~F. Kammerer, J.~Lu, and A.~Watson}, {\em
  Edge state dynamics along curved interfaces}, arXiv:2106.00729,  (2021).

\bibitem{bal2022multiscale}
{\sc G.~Bal and D.~Massatt}, {\em Multiscale invariants of {Floquet}
  topological insulators}, Multiscale Modeling \& Simulation, 20 (2022),
  pp.~493--523.

\bibitem{bellissard1994noncommutative}
{\sc J.~Bellissard, A.~van Elst, and H.~Schulz-Baldes}, {\em {The
  noncommutative geometry of the quantum Hall effect}}, Journal of Mathematical
  Physics, 35 (1994), pp.~5373--5451.

\bibitem{bernevig2013topological}
{\sc B.~A. Bernevig and T.~L. Hughes}, {\em Topological insulators and
  topological superconductors}, Princeton university press, 2013.

\bibitem{bolte2004semiclassical}
{\sc J.~Bolte and R.~Glaser}, {\em {A semiclassical Egorov theorem and quantum
  ergodicity for matrix valued operators}}, Communications in mathematical
  physics, 247 (2004), pp.~391--419.

\bibitem{bony1996caracterisations}
{\sc J.-M. Bony}, {\em Caract{\'e}risations des op{\'e}rateurs
  pseudo-diff{\'e}rentiels}, S{\'e}minaire {\'E}quations aux d{\'e}riv{\'e}es
  partielles (Polytechnique),  (1996), pp.~1--15.

\bibitem{bony2013characterization}
\leavevmode\vrule height 2pt depth -1.6pt width 23pt, {\em On the
  characterization of pseudodifferential operators (old and new)}, in Studies
  in Phase Space Analysis with Applications to PDEs, Springer, 2013,
  pp.~21--34.

\bibitem{bony1994espaces}
{\sc J.-M. Bony and J.-Y. Chemin}, {\em Espaces fonctionnels associ{\'e}s au
  calcul de weyl-h{\"o}rmander}, Bulletin de la soci{\'e}t{\'e}
  Math{\'e}matique de France, 122 (1994), pp.~77--118.

\bibitem{bony1989quantification}
{\sc J.-M. Bony and N.~Lerner}, {\em {Quantification asymptotique et
  microlocalisations d'ordre sup{\'e}rieur. I}}, in Annales scientifiques de
  l'Ecole normale sup{\'e}rieure, vol.~22, 1989, pp.~377--433.

\bibitem{bott1978some}
{\sc R.~Bott and R.~Seeley}, {\em Some remarks on the paper of {Callias}},
  Communications in Mathematical Physics, 62 (1978), pp.~235--245.

\bibitem{bourne2017k}
{\sc C.~Bourne, J.~Kellendonk, and A.~Rennie}, {\em The k-theoretic bulk--edge
  correspondence for topological insulators}, in Annales Henri Poincar{\'e},
  vol.~18, Springer, 2017, pp.~1833--1866.

\bibitem{bourne2018chern}
{\sc C.~Bourne and A.~Rennie}, {\em Chern numbers, localisation and the
  bulk-edge correspondence for continuous models of topological phases},
  Mathematical Physics, Analysis and Geometry, 21 (2018), p.~16.

\bibitem{callias1978axial}
{\sc C.~Callias}, {\em Axial anomalies and index theorems on open spaces},
  Communications in Mathematical Physics, 62 (1978), pp.~213--234.

\bibitem{combes2005edge}
{\sc J.-M. Combes and F.~Germinet}, {\em Edge and impurity effects on
  quantization of hall currents}, Communications in mathematical physics, 256
  (2005), pp.~159--180.

\bibitem{davies_1995}
{\sc E.~B. Davies}, {\em Spectral Theory and Differential Operators}, Cambridge
  Studies in Advanced Mathematics, Cambridge University Press, 1995.

\bibitem{delplace2017topological}
{\sc P.~Delplace, J.~Marston, and A.~Venaille}, {\em Topological origin of
  equatorial waves}, Science, 358 (2017), pp.~1075--1077.

\bibitem{dimassi1999spectral}
{\sc M.~Dimassi and J.~Sj{\"o}strand}, {\em Spectral asymptotics in the
  semi-classical limit}, no.~268, Cambridge university press, 1999.

\bibitem{dombrowski2011quantization}
{\sc N.~Dombrowski, F.~Germinet, and G.~Raikov}, {\em Quantization of edge
  currents along magnetic barriers and magnetic guides}, in Annales Henri
  Poincar{\'e}, vol.~12, Springer, 2011, pp.~1169--1197.

\bibitem{drouot2021microlocal}
{\sc A.~Drouot}, {\em Microlocal analysis of the bulk-edge correspondence},
  Communications in Mathematical Physics, 383 (2021), p.~2069–2112.

\bibitem{DFN-SP-1985}
{\sc B.~A. Dubrovin, A.~T. Fomenko, and S.~P. Novikov}, {\em Modern
  geometry---methods and applications. {P}art {II}: The Geometry and Topology
  of Manifolds}, Springer-Verlag, New York, 1985.

\bibitem{elbau2002equality}
{\sc P.~Elbau and G.~Graf}, {\em Equality of bulk and edge hall conductance
  revisited}, Communications in mathematical physics, 229 (2002), pp.~415--432.

\bibitem{essin2011bulk}
{\sc A.~M. Essin and V.~Gurarie}, {\em Bulk-boundary correspondence of
  topological insulators from their respective green’s functions}, Physical
  Review B, 84 (2011), p.~125132.

\bibitem{fedosov1970direct}
{\sc B.~V. Fedosov}, {\em Direct proof of the formula for the index of an
  elliptic system in euclidean space}, Functional Analysis and Its
  Applications, 4 (1970), pp.~339--341.

\bibitem{FLW-ES-2015}
{\sc C.~L. Fefferman, J.~P. Lee-Thorp, and M.~I. Weinstein}, {\em Edge states
  in honeycomb structures}, Annals of PDE, 2 (2016), p.~12.

\bibitem{fukui2012bulk}
{\sc T.~Fukui, K.~Shiozaki, T.~Fujiwara, and S.~Fujimoto}, {\em {Bulk-edge
  correspondence for Chern topological phases: A viewpoint from a generalized
  index theorem}}, Journal of the Physical Society of Japan, 81 (2012),
  p.~114602.

\bibitem{graf2007aspects}
{\sc G.~M. Graf}, {\em Aspects of the integer quantum hall effect}, in
  Proceedings of Symposia in Pure Mathematics, vol.~76, Providence, RI;
  American Mathematical Society; 1998, 2007, p.~429.

\bibitem{graf2020topology}
{\sc G.~M. Graf, H.~Jud, and C.~Tauber}, {\em Topology in shallow-water waves:
  a violation of bulk-edge correspondence}, arXiv preprint arXiv:2001.00439,
  (2020).

\bibitem{Graf2013}
{\sc G.~M. Graf and M.~Porta}, {\em Bulk-edge correspondence for
  two-dimensional topological insulators}, Communications in Mathematical
  Physics, 324 (2013), pp.~851--895.

\bibitem{hatsugai1993chern}
{\sc Y.~Hatsugai}, {\em Chern number and edge states in the integer quantum
  hall effect}, Physical review letters, 71 (1993), p.~3697.

\bibitem{H-III-SP-94}
{\sc L.~V. H{\"o}rmander}, {\em {The Analysis of Linear Partial Differential
  Operators III: Pseudo-Differential Operators}}, Springer Verlag, 1994.

\bibitem{liu2010model}
{\sc C.-X. Liu, X.-L. Qi, H.~Zhang, X.~Dai, Z.~Fang, and S.-C. Zhang}, {\em
  Model hamiltonian for topological insulators}, Physical Review B, 82 (2010),
  p.~045122.

\bibitem{ludewig2020cobordism}
{\sc M.~Ludewig and G.~C. Thiang}, {\em Cobordism invariance of topological
  edge-following states}, arXiv preprint arXiv:2001.08339,  (2020).

\bibitem{nicola2011global}
{\sc F.~Nicola and L.~Rodino}, {\em {Global pseudo-differential calculus on
  Euclidean spaces}}, vol.~4, Springer Science \& Business Media, 2011.

\bibitem{nirenberg1974topics}
{\sc L.~Nirenberg}, {\em Topics in nonlinear functional analysis}, vol.~6,
  American Mathematical Soc., 1974.

\bibitem{prodan2016bulk}
{\sc E.~Prodan and H.~Schulz-Baldes}, {\em {Bulk and boundary invariants for
  complex topological insulators: From K-Theory to Physics}}, Springer Verlag,
  Berlin, 2016.

\bibitem{RevModPhys.83.1057}
{\sc X.-L. Qi and S.-C. Zhang}, {\em Topological insulators and
  superconductors}, Rev. Mod. Phys., 83 (2011), pp.~1057--1110.

\bibitem{QB-NUMTI-2021}
{\sc S.~Quinn and G.~Bal}, {\em {Approximations of interface topological
  invariants}}, arXiv:2112.02686,  (2022).

\bibitem{sato2017topological}
{\sc M.~Sato and Y.~Ando}, {\em Topological superconductors: a review}, Reports
  on Progress in Physics, 80 (2017), p.~076501.

\bibitem{schindler2018higher}
{\sc F.~Schindler, A.~M. Cook, M.~G. Vergniory, Z.~Wang, S.~S. Parkin, B.~A.
  Bernevig, and T.~Neupert}, {\em Higher-order topological insulators}, Science
  advances, 4 (2018), p.~eaat0346.

\bibitem{schulz2000simultaneous}
{\sc H.~Schulz-Baldes, J.~Kellendonk, and T.~Richter}, {\em Simultaneous
  quantization of edge and bulk hall conductivity}, Journal of Physics A:
  Mathematical and General, 33 (2000), p.~L27.

\bibitem{souslov2019topological}
{\sc A.~Souslov, K.~Dasbiswas, M.~Fruchart, S.~Vaikuntanathan, and V.~Vitelli},
  {\em Topological waves in fluids with odd viscosity}, Physical Review
  Letters, 122 (2019), p.~128001.

\bibitem{tauber2019bulk}
{\sc C.~Tauber, P.~Delplace, and A.~Venaille}, {\em A bulk-interface
  correspondence for equatorial waves}, Journal of Fluid Mechanics, 868 (2019).

\bibitem{volovik2009universe}
{\sc G.~Volovik}, {\em The Universe in a Helium Droplet}, International Series
  of Monographs on Physics, OUP Oxford, 2009.

\bibitem{witten2016three}
{\sc E.~Witten}, {\em Three lectures on topological phases of matter}, La
  Rivista del Nuovo Cimento, 39 (2016), pp.~313--370.

\bibitem{zworski2012semiclassical}
{\sc M.~Zworski}, {\em Semiclassical analysis}, vol.~138, American Mathematical
  Soc., 2012.

\end{thebibliography}
%\bibliographystyle{siam}

}

\appendix

\section{Notation, operators and functional calculus}
\label{sec:notation} 
%%
%
%\gb{Should go to an appendix. Cite book Zworksi for Hilbert spaces. }

This appendix recalls results summarized in \cite{bony2013characterization} allowing us to characterize spaces of symbols $a_k$ for $H_k=\ow a_k$ adapted to operators modeling unbounded domain walls, domains of definition for $H_k$, as well as functional calculus showing in particular that $(z-H_{d-1})^{-1}$ and $\varphi'(H_{d-1})$ are pseudo-differential operators. We also recall results on semiclassical calculus and the Helffer-Sj\tio strand formula following \cite{dimassi1999spectral}.

\medskip\noindent{\bf Symbol spaces \cite{bony2013characterization}.} See also \cite[Chapters 18\&19]{H-III-SP-94},  \cite{bony1996caracterisations,bony1994espaces,bony1989quantification} and \cite[Chapter 8.3]{zworski2012semiclassical}.

On phase space $\Rm^{2d}$ in $d$ spatial dimensions parametrized by $X=(x,\xi)$ with $x\in \Rm^d$ and $\xi\in  (\Rm^d)^*\simeq \Rm^d$, we define a Riemannian metric $g$ in the  Beals-Fefferman form by
\[
  g_X(dx,d\xi) = \Phi_x^{-2}(X) dx^2 + \Phi_\xi^{-2}(X) d\xi^2.
\]
We use the notation $\aver{u}=\sqrt{1+|u|^2}$ for $|\cdot|$ the Euclidean norm applied to a vector $u$. For $u=(u_1,u_2)$, we use the notation $\aver{u_1,u_2}=\sqrt{1+|u_1|^2+|u_2|^2}$. Associated to the above metric $g$, we define the Planck function $h(X)$ and its inverse $\lambda(X)$ by $1\leq h^{-1}(X)=\lambda(X)=\Phi_\xi(X)\Phi_x(X)$. 

In this paper, we consider two metrics: $g^i$ and $g^s$. The metric $g^i$ is defined by $\Phi_x^i(X)=\Phi_\xi^i(X)=\aver{X}\geq1$ with then $h^i=\aver{x,\xi}^{-2}$. The metric $g^s$ is defined by $\Phi_x^s(X)=\aver{x}\geq1$ and $\Phi_\xi^s(X)=\aver{\xi}\geq1$ with Planck function $h^s=\aver{x}^{-1}\aver{\xi}^{-1}$.

We define a weight (order function) $M$ and a class of symbols $S(M,g)$ such that  $M\in S(M,g)$ is {\em admissible} for $g$; see \cite[Definition 2.3]{bony2013characterization}. For $0\leq k\leq d$, we decompose $x=(x_k',x_k'')$ with  $x_k'\in \Rm^{k}$ and $x_k''\in\Rm^{d-k}$. We define the weights 
\begin{equation}\label{eq:weights}
w_k(X)=\aver{x_k',\xi} \qquad \mbox{ and }\qquad M_k(X)=w_k^m(X),\quad \mbox{ for } \quad m\in \Nm_+.
\end{equation}
Following classical calculations \cite{bony2013characterization}, the weights $M_k$ are admissible for $g\in \{g^i,g^s\}$ and satisfy that $M_k\leq C\lambda^p$ for some $C>0$ and $p<\infty$ when $\lambda\in\{\lambda^i,\lambda^s\}$. This implies in particular that $M_kh^N$ goes to $0$ as $X\to\infty$ for $N$ sufficiently large when $h\in \{h^i,h^s\}$.

For $g=g^s$, a symbol $b\in S(M,g^s)$ precisely when \eqref{eq:Sk} holds \cite{bony2013characterization,H-III-SP-94}. For $g=g^i$, a symbol $b\in S(M,g^i)$ when \eqref{eq:Sk} holds with $\aver{x}^{|\alpha|}\aver{\xi}^{|\beta|}$ replaced by $\aver{X}^{|\alpha|+|\beta|}$. The metric $g^i$ is referred to as the isotropic metric with $S(M,g^i)$ a subspace of $S(M,g^s)$.  Since they appear repeatedly in the derivations, we define for $0\leq k\leq d$  the spaces:
\begin{equation}\label{eq:symbols}
  S_k^m(g)= S(w_k^m,g)\otimes \Mm(n_k) \qquad \mbox{ and } \qquad \tilde S^m(g)= S(w^m_d,g)\otimes \Mm(n_{d-1}).
\end{equation}
We also define $S_k^m:=S_k^m(g^s)$ and $\tilde S^m:= \tilde S^m(g^s)$. Here, $n_k$ and $n_{d-1}$ are the dimensions of the spinors given in the introduction and in the construction of the augmented Hamiltonians in section \ref{sec:local} while $\Mm(n)$ is the space of $n\times n$ matrices with complex coefficients.

\medskip\noindent{\bf Hilbert Spaces \cite[\S2.2.3]{bony2013characterization}\cite{bony1994espaces}.} See also  \cite{bony1996caracterisations,bony1989quantification,nicola2011global} and \cite[Chapter 8.3]{zworski2012semiclassical}.

Associated to the weights $M_k$ are the Hilbert spaces $\mH(M_k,g)$ of $u\in\mS'(\Rm^d)$ such that $\ow a \,u\in L^2(\Rm^d)$ whenever $a\in S(M_k,g)$. These spaces are in fact independent of $g\in \{g^s,g^i\}$ \cite{bony2013characterization} and hence referred to as $\mH(M_k)$. We observe that $\mH(1,g)= L^2(\Rm^d)$. The Hilbert spaces associated to $S_k^m$ and $\tilde S^m$ are denoted for $0\leq k\leq d$ by
\begin{equation}\label{eq:Hspaces}
  \mH_k^m= \mH(w_k^m) \otimes \Mm(n_k)\qquad \mbox{ and } \qquad\tilde \mH^m= \mH(w_d^m)\otimes \Mm(n_{d-1}).
\end{equation}
When $a_k\in S^m_k(g)$, we thus obtain that $\ow a_k$ maps $\mH^m_k$ to $\mH^0_k = L^2(\Rm^d)\otimes \Mm(n_k)$. Note that $\mH^0_{d-1}=\tilde \mH^0=L^2(\Rm^d)\otimes \Mm(n_{d-1})$. Pseudo-differential operators with Weyl quantization are defined in \eqref{eq:weyl} with integrals defined as oscillatory integrals.

\medskip\noindent{\bf Ellipticity \cite[\S 2.3.3]{bony2013characterization}.}
We say that $a\in S(M,g)\otimes \Mm(n)$ Hermitian valued is {\bf elliptic} when
\[
  |{\rm det}\ a(X)| ^{\frac1n} \geq C_1M(X)- C_2
\]
for some positive constants $C_{1,2}$. This is equivalent to imposing that each eigenvalue of $a(X)$ is bounded away from $0$ by at least $CM(X)$ outside of a compact set.  We then say that $a\in ES(M,g)\otimes \Mm(n)$ and define the corresponding spaces of Hermitian elliptic symbols as $ES^m_k(g)$ for $0\leq k\leq d$ and $E\tilde S^m(g)$.

Since $M(X)\leq C\lambda^p(X)$, ellipticity implies that $H=\ow a \in \ow ES(M,g)\otimes \Mm(n)$ is a self-adjoint operator with domain of definition $\mD(H)=\mH(M)\otimes \Mm(n)$ and such that for some positive constant $C$ \cite{bony2013characterization}
\[
  \|H\psi \|_{L^2(\Rm^d)\otimes \Mm(n)} \geq C ( \|\psi\|_{\mH(M)\otimes \Mm(n)} - \|\psi\|_{L^2(\Rm^d)\otimes \Mm(n)}).
\]
For $a_k\in ES^m_k$, we thus obtain that $H_k=\ow a_k$ is a self-adjoint operator from its domain of definition $\mD(H_k)=\mH^m_k$ to $\mH^0_k$. Similarly, for $a\in E\tilde S^m$, then $F=\ow a$  is an unbounded operator from its domain of definition $\mD(F)=\tilde\mH^m$ to $\tilde\mH^0$.

\paragraph{Functional calculus \cite[\S 2.3]{bony2013characterization}.}
For $H=\ow a$ and $a$ elliptic, the above results show  that the resolvent $(z-H)^{-1}$ is an isomorphism from $L^2(\Rm^d)\otimes \Mm(n)$ to $\mH(M)\otimes \Mm(n)$ for $z\in\Cm$ when ${\rm Im}(z)\not=0$.

With the above assumptions, we have the {\em Wiener property} \cite{bony2013characterization} stating that: (i) $A\in \ow S(1,g)$ invertible in $\mL(L^2)$ implies that $A^{-1}\in \ow S(1,g)$; and (ii) $A\in \ow S(M,g)$ bijection from $\mH(M_1,g)$ to $\mH(M_1/M,g)$, then $A^{-1}\in \ow S(M^{-1},g)$.

This allows us to state the following result:
\begin{lemma}\label{lem:resiso}
 Let $H\in \ow ES_{d-1}^m$. Then  %$(\pm i+H)^{-1}\in \ow S(m_{d-1}^{-1})$. \tb{Has to be rewritten as:} 
 $(\pm i+H)^{-1}\in \ow ES_{d-1}^{-m}(g)$ is an isomorphism from  $\mH_{d-1}^0$ to $\mH_{d-1}^m$.
\end{lemma}
\begin{proof}
The proof follows \cite{bony1996caracterisations,bony2013characterization}.  Associated to $H$ is a resolvent operator $R_z=(z-H)^{-1}$, which is always defined as a bounded operator by spectral theory. %(spectral theory can be constructed using the resolvent; see Tao's notes). 
When $H$ is elliptic, then the domain $\mD(H)=\mH^m_{d-1}$. Moreover, $R_z$ is a bijection from $\mH^0_{d-1}=L^2(\Rm^n)\otimes\Mm(n_{d-1})$ to that domain. We then apply above Wiener property \cite{bony1996caracterisations,bony2013characterization,bony1994espaces,bony1989quantification} to obtain that $R_z^{-1}\in\ow S^{-m}_{d-1}(g)$. 
\end{proof}
The above shows that $(I+H^2)^{-1}$ maps $\mH_{d-1}^0$ to $\mH_{d-1}^{2m}$ and has a symbol in $ES_{d-1}^{-2m}$. Moreover, using the Helffer-Sj\tio strand formula as done in \cite[Theorem 4]{bony2013characterization} using $p\to-\infty$ in the notation there, we obtain the following result on the functional calculus:
\begin{lemma}\label{lem:fccalc}
  Let $\phi\in C^\infty_c(\Rm)$ and $H\in \ow ES_{d-1}^m$. Then $\phi(H)\in \ow S_{d-1}^{-\infty}$.
\end{lemma}
%\begin{proof}
%Bony's paper with $p\to-\infty$.
%\end{proof}
\begin{remark}\label{rem:unifbd}
  The above result means the following in terms of seminorms: for each $N\in\Nm$ and each seminorm $k$ defining the topology on the space of symbols, there is a seminorm $l$ such that $\phi(H)$ is bounded for the seminorm $k$ uniformly in the seminorm $l$ applied to $a$. 
 %This cannot be a linear functional though since $\phi$ is nonlinear. 
 For a sequence of operators $H(\eps)= \ow a(\eps)$ with $a(\eps)$ with seminorms of $S^m_{d-1}$ uniformly bounded in $\eps$, this implies that the symbol of $\phi(H_\eps)$ is bounded in any $S^{-N}_{d-1}$ uniformly in $\eps$ as well.
\end{remark} 

%\tb{This concludes the results on symbols and functional calculus. We probably want to present here the results we need on Moyal product and semiclassical expansions.}

\medskip
\noindent{\bf Semiclassical calculus \cite{dimassi1999spectral}.} The computation of several topological invariants, as in the proof of \cite[Theorem 19.3.1]{H-III-SP-94}, simplifies in the semiclassical regime. Let $0<h\leq h_0\leq 1$ be the semiclassical parameter. We define semiclassical operators in the Weyl quantization as 
\begin{equation}\label{eq:weylh}
  H_h \psi= \ow_h a\  \psi := \dfrac{1}{(2\pi h)^d} \dint_{\Rm^{2d}} e^{i\frac1h\xi\cdot (x-y)} a(\frac{x+y}2,\xi;h)  \psi(y)dy d\xi,
\end{equation}
for $a(X;h)$ a matrix-valued symbol in $\Mm(n)$ for each $X\in\Rm^{2d}$ and $h\in(0,h_0]$ and $\psi(x)$ a spinor in $\Cm^n$.  The semi-classical symbol $a(X;h)$ is related to the Schwartz kernel $K(x,y;h)$ of $H_h$ as
\[
 a(x,\xi;h) = \dint_{\Rm^{d}} e^{-i\frac1h \xi\cdot y} K(x+\frac y2,x-\frac y2;h) dy.
\]
Note that $\ow a(x,h\xi;h)=\ow_h a(x,\xi;h)$. We define the classes of semi-classical symbols as $\rS^j(M)$ constructed with the semi-classical metric in Beals-Fefferman form with $\Phi_x(X)=1$ and $\Phi_\xi(X)=h^{-1}$, and for $M$ an order function, i.e., in this context a non-negative function on $\Rm^{2d}$ satisfying $M(x,\xi)\leq C(1+|x-y|+|\xi-\zeta|)^N M(y,\zeta)$ uniformly in $(x,y,\xi,\zeta)$ for some $C(M)$ and $N(M)$. Then $a\in \rS^j(M)\otimes \Mm(n)$ if for each component $b$ of $a$, we have for each $2d-$dimensional multi-index $\alpha$, a constant $C_\alpha$ such that
\begin{equation}\label{eq:Sjh}
   h^{-j}\,  | \partial^\alpha_X b(X;h) | \leq C_\alpha   M(X),\quad \forall X\in\Rm^{2d},\ \forall h\in (0,h_0].
\end{equation}
We will mostly use the case $j=0$. 

For two operators $\ow_h a$ and $\ow_h b$ with symbols $a\in S^0(M_1)$ and $b\in S^0(M_2)$, we then define the composition $\ow_h c = \ow_h a \ow_h b$ with symbol $c\in S^0(M_1M_2)$ given by the (Moyal) product \cite[Thm. 7.9]{dimassi1999spectral}
\begin{equation}\label{eq:sharph}
c(x,\xi) =(a\sharp_h b) (x,\xi) := \Big( e^{i\frac h2(\partial_x\cdot\partial_\zeta - \partial_y\cdot\partial_\xi)} a(x,\xi) b(y,\zeta)\Big)_{|y=x;\zeta=\xi}.
\end{equation}

For $a\in S^0(1)$, we obtain (\cite[Thm. 7.11]{dimassi1999spectral},\cite[Prop. 1.4]{bolte2004semiclassical}) that $\ow_h a$ is bounded as an operator in ${\mathcal L}(L^2(\Rm^d)\otimes \Cm^n)$ with bound uniform in $0<h\leq h_0$ so that $I-h\ow_h a$ is invertible on that space when $h$ is sufficiently small.

An operator is said to be semiclassically elliptic when the symbol $a=a(x,\xi;h)\in S^0(M)$ is invertible in $\Mm_n$ for all $(x,\xi)\in \Rm^{2d}$ and $h\in (0,h_0]$ with then $a^{-1}\in S^0(m^{-1})$. 

Following \cite{dimassi1999spectral} (see \cite[Lemma 4.14]{bal2022topological}), we obtain the following results on resolvent operators. Let $H_h=\ow_h a$ with $a\in S^0(m)$. Let $z=\lambda+i\omega\in\Cm$ with $\omega\not=0$. Then $(z-H_h)^{-1}$ is a bounded operator and there exists an analytic function $z\to r_z=r_z(y,\zeta;h)$ such that $(z-H_h)^{-1}=\ow_h r_z$ (compare to Lemma \ref{lem:resiso}). Moreover, the symbol $r_z\in S^0(1)$ satisfies
\begin{equation}\label{eq:resbd}
  |\partial^\beta_{(y,\eta)} r_z| \leq C_\beta | \omega|^{-1-|\beta|} (1+h^{\frac{2d+1}2}|\omega|^{-2d-1}),
\end{equation}
for all multi-indices $\beta=(\beta_y,\beta_\zeta)$ and a constant $C_\beta$ independent of $z\in Z$ a compact set in $\Cm$ and $0<h\leq h_0$.

\medskip
\noindent{\bf Helffer-Sj\tio strand formula \cite{davies_1995,dimassi1999spectral}.} Finally, we recall some results on spectral calculus and the Helffer-Sj\"ostrand formula following \cite{davies_1995,dimassi1999spectral}; see also \cite{bolte2004semiclassical} for the vectorial case. For any self-adjoint operator $H$ from its domain $\mD(H)$ to $L^2(\Rm^d)\otimes\Cm^n$ and any bounded continuous function $\phi$ on $\Rm$, then $\phi(H)$ is uniquely defined as a bounded operator on $L^2(\Rm^d)\otimes\Cm^n$ \cite[Chapter 4]{dimassi1999spectral}. Moreover, for $\phi$ compactly supported, we have the following representation
\begin{equation}\label{eq:hs}
  \phi(H) = -\frac1\pi \dint_{\Cm} \bar\partial \tilde\phi(z) (z-H)^{-1} d^2z,
\end{equation}
where, for $z=\lambda+i\omega$, $d^2z:=d\lambda d\omega$, $\bar\partial=\frac12\partial_\lambda+\frac1 2\partial_\omega$, and where $\tilde\phi(z)$ is an almost analytic extension of $\phi$.
The extension $\tilde\phi$ is compactly supported in $\Cm$.  Moreover, $\tilde\phi(\lambda+i0)=\phi(\lambda)$ and $\bar\partial \tilde\phi(\lambda+i0)=0$, whence the name of {\em almost} analytic extension.  We can in fact choose the almost analytic extension such that $|\bar\partial\tilde\phi| \leq C_N |\omega|^N$  for any $N\in\Nm$ in the vicinity of the real axis uniformly in $(\lambda,\omega)$ on compact sets.    Several explicit expressions, which we do not need here, for such extensions are available in \cite{davies_1995,dimassi1999spectral}.

\end{document}